\documentclass[reqno]{amsart}

\usepackage{graphicx}
\usepackage{esint}

\newtheorem{lemma}{Lemma}
\newtheorem{theorem}{Theorem}

\newtheorem{remark}{Remark}
\newtheorem{definition}{Definition}
\newcommand{\be}{\begin{eqnarray}}
\newcommand{\ee}{\end{eqnarray}}
\newcommand{\bee}{\begin{eqnarray*}}
\newcommand{\eee}{\end{eqnarray*}}
\newcommand{\R}{{\mathbb R}}
\newcommand{\C}{{\mathbb C}}
\newcommand{\N}{{\mathbb N}}
\newcommand{\Z}{{\mathbb Z}}

\newcommand{\asy}{{\mathcal O}}

\begin{document}

\title [Existence of the Stark-Wannier quantum resonances] {Existence of the Stark-Wannier quantum resonances}

\author {Andrea SACCHETTI}

\address {Department of Physics, Computer Sciences and Mathematics, University of Modena e Reggio Emilia, Modena, Italy.}

\email {andrea.sacchetti@unimore.it}

\date {\today}

\thanks {This paper is dedicated to Professor Vincenzo Grecchi on the occasion of his 70th birthday}

\begin {abstract} In this paper we prove the existence of the Stark-Wannier quantum resonances for one-dimensional Schr\"odinger 
operators with smooth periodic potential and small external homogeneous electric field. \ Such a result extends the existence result 
previously obtained in the case of periodic potentials with a finite number of open gaps. 
\bigskip

{\it Ams classification (MSC 2010):} 81Qxx, 81Q15. 

\bigskip


\end{abstract}

\maketitle

\section {Introduction} Existence of quantum resonances for the one-dimensional Schr\"odinger operator with a periodic potential and 
an external homogeneous electric field has been a largely debated question among the solid state community since the seminal papers by 
Wannier \cite {W1,W2}. \ These resonances, named Stark-Wannier resonances, are strictly connected with the time behavior of the Bloch 
oscillators \cite {GKK}, which electron wave function $\psi (x,t)$ is the solution of the time-dependent Schr\"odinger equation 
\bee
i \hbar \frac {\partial \psi}{\partial t} = H_F \psi 
\eee
where $H_F$ is the Stark-Wannier operator formally defined on $L^2 (\R , dx)$ as
\bee
H_F := H_B + F x, \ H_B:=  - \frac {\hbar^2}{2m}\frac {d^2}{dx^2} + V(x) \, , 
\eee
$Fx $ represents the Stark potential associated to the external homogeneous electric field, $H_B$ is the Bloch operator and $V(x)$ represents 
the periodic potential of an one-dimensional crystal; here we assume to choose the units such that $2m =1$ and $\hbar =1$ is fixed. 

Assuming that $V(x)$ is regular enough then the spectrum of $H_F$ covers the whole real axis. \ On the other side, if we neglect the 
interband coupling term then it turns out that the spectrum of such a decoupled band approximation consists of a sequence on infinite 
ladders of real eigenvalues. \ The crucial point is to understand what happen to these eigenvalues when we restore the interband coupling 
terms. \ This question has been largely debated \cite {W3,Z1,Z2} and at present it is expected that these ladders of real eigenvalues 
will turn into quantum resonances.

Despite its apparently simplicity this problem is rather technically complicated to study because the external potential $F x$ is 
not a bounded perturbation and so the interband coupling term is quite hard to manage.

A first result in order to put on a rigorous ground this problem was due to Herbst and Howland \cite {HH}, who gave a rigorous 
definition for quantum resonances for the Wannier-Stark operator by means of the complex translation method (in fact, this 
definition applies to the case of analytic periodic potential); that is a quantum resonance corresponds to a complex-valued 
eigenvalue of a non self-adjoint operator obtained by $H_F$ by means of the analytic translation. \ More precisely, quantum resonances 
for the Stark-Wannier operator are defined by Herbst and Howland as the complex eigenvalues of the analytically translated operator
\bee
H_F^\theta = {\mathcal U}^\theta H_F \left ( {\mathcal U}^\theta \right )^{-1} = - \frac {d^2}{dx^2} + V(x+\theta ) + 
F x + F \theta 
\eee
where $-R < \Im \theta < 0$, the periodic potential was assumed to be analytic in a set containing the strip $|\Im z | \le R$, and  
\bee
{\mathcal U}^\theta : L^2 (\R , dx) \to L^2 (\R , dx )
\eee
is the translation operator defined as 
\bee
\left ( {\mathcal U}^\theta \psi \right ) (x) = \psi (x +\theta ) \, .
\eee
In particular, they proved that 
\bee
\sigma_{ess} (H_F^\theta ) = \left \{ z \ :\ \Im z = F \Im \theta \right \} \, , \ 
\sigma_d (H_F^\theta) \subseteq \{ z \ : \ F \Im \theta < \Im z \le 0 \} 
\eee
and, let 
\bee
{\mathcal R} = \cup_{\theta \ : \ 0 < -\Im \theta < R} \sigma_d (H_F^{\theta} ) \, , 
\eee
then the kernel  of the resolvent operator $[z-H_F]^{-1}$ has an analytic continuation from $\Im z >0$ to $\{ z \ : \ \Im z > - R \} 
\setminus {\mathcal R}$, and if $z_0 \in {\mathcal R}$ then the kernel of the resolvent operator $[z_0 - H_F ]^{-1}$ 
has a pole at this point. \ The points of ${\mathcal R}$ with strictly negative imaginary part are called quantum resonances.

A second important result was obtained by proving that the spectrum of the Stark-Wannier operator is purely absolutely 
continuous \cite {BCDSSW}
\bee
\sigma (H_F) = \sigma_{ac} (H_F) = \R \, ;
\eee
hence, embedded eigenvalues are not admissible. \ By means of this result we can so exclude the possibility that the ladders of real 
eigenvalues for the decoupled band approximation will turn into real embedded eigenvalues when the coupling interband term is restored. \ In 
fact, this is not the case of singular periodic potentials; in particular, when the periodic potential is given by means of a periodic 
sequence of Dirac's $\delta$ pointwise interactions then we expect to observe embedded eigenvalues \cite {FGZ}, as well as in the case 
of a periodic sequence of Dirac's $\delta$-prime pointwise interactions \cite {MS}.

Finally, the proof of existence of ladders of quantum resonances was given in the limit of small electric field and under a 
technical assumption: it had been assumed that the number of open gaps is finite \cite {GMS}. \ We should recall that, in 3-dimensional 
crystals, the number of open gaps is usually finite; this is not the case of 1-dimensional crystals where we have a finite 
number of open gaps only for periodic potentials given by a class of elliptic functions.

In this paper we are able to remove such a technical assumption on the periodic potential, that is we prove the existence of 
quantum resonances for electric field small enough and for any periodic potential analytic in a finite strip containing the 
real axis. \ In our proof we develop some ideas already contained in \cite {GMS}, the main improvement is that we are able to 
manage here the interband coupling term (see Theorem \ref {Theorem3} and Theorem \ref {Theorem3Bis}), even in the case of infinitely many open gaps, 
by means of precise estimates of the asymptotic behavior of the Bloch functions around the Kohn branch points.

In our treatment we replace also the usual definition of quantum resonances, given by Herbst and Howland, by considering the complex 
eigenvalues of the non self-adjoint operator obtained by $H_F$ in the \emph {crystal momentun representation} (CMR) by means 
of the analytic distortion $\hat \psi (p) \to e^{i p \theta } \hat \psi (p)$ for $p$ large enough and $\Im \theta <0$; 
where $\hat \psi$ represents the wave-function in the crystal momentum representation and $p$ is the crystal momentum variable.

We close by mentioning the fact that the problem of the existence of the Stark-Wannier quantum resonances and the time behavior 
of Bloch oscillators have been the object of a large analysis, where many different techniques have been adopted; among all the papers 
we recall the works by Nenciu \cite {N1,N2} and H\"overmann, Spohn and Teufel \cite {HST}, where adiabatic methods has been applied, the works by Buslaev and 
Dmitrieva \cite {BD}, where a WKB-type expansion has been adapted in order to match the asymptotic behaviour of the wavefunction in a 
neighbourhood of the turning points in the tilted band picture, the works by Bentosela and Grecchi \cite {BG} and by Combes 
and Hislop \cite {CH}, where the problem of the existence of the Stark-Wannier resonances has been considered in the semiclassical 
limit $\hbar \to 0$, and finally we mention the paper by Avron \cite {Av}, where he considered the spectral problem in the case of 
complex-valued $F$.

The paper is organized as follows.

In \S 2 we state our main result.

In \S 3 we collect some known results about band and Bloch functions.

In \S 4 we give the asymptotic behavior of the Bloch functions in the high energy limit.

In \S 5 we adapt the CMR to the analytic structure of the Bloch functions, the obtained representation is named \emph {extended crystal 
momentum representation} (ECMR), and then we consider the Stark-Wannier operator $H_F$ in the ECMR.

In \S 6 we define the analytic distortion in the ECMR.

In \S 7 we study the spectral properties of the decoupled band approximation of the analytically distorted Stark-Wannier operator in the ECMR.

In \S 8 we prove that the coupling term between the bands is a bounded operator and then we define quantum resonances as complex valued 
eigenvalues of the analytically distorted Stark-Wannier operator in the ECMR. \ We then restore the coupling term and we prove, by means of the regular perturbation theory, that the 
ladders of eigenvalues of the decoupled band approximation turns into complex eigenvalues of the distorted Stark-Wannier operator contained 
in the strip $|\Im z| \le \frac 12 F |\Im \theta |$. Furthermore, we prove that the essential spectrum of the distorted Stark-Wannier operator 
lies below the line $\frac 12 F \Im \theta $, where $\Im \theta <0$. \ In such a way the existence of Stark-Wannier quantum resonances will 
follow.

\clearpage

\subsection* {Notation} In the following:

\begin {itemize}

\item [-] we denote by $E_n (k)$ and $\varphi_n (x,k)$ the band and Bloch functions where $n \in \N$ and $k$ is the quasi-momentun 
(or crystal momentum) variable belonging to the Brillouin zone ${\mathcal B} = \left ( - \frac 12 b , + \frac 12 b \right ]$, 
$b= \frac {2\pi}{a}$ where $a$ is the period of the periodic potential;

\item [-] we denote by $E(p)$ and 
\bee
\varphi (x,p) = e^{ipx} u(x,p) 
\eee
the multisheeted band and Bloch functions, where the quasimomentum $p$ belongs to the complex plane $\tilde \C$ with cuts connecting the 
Kohn branch points $k_n = \pm \frac 12 n b \pm i \kappa_n$, $\kappa_n \ge 0$ (actually, if $\kappa_n =0$ then the gap is closed and there 
is no branch point);

\item [-] since we have to work with meromorphic functions $f(p)$ defined on $\tilde \C$ then, by construction, $f(p)$ 
would have discontinuity points on $\R \setminus \tilde \R$, where $\tilde \R = \R \cap \tilde \C$; we'll denote by $\fint_A f(p) dp$, where $A$ in an interval in $\R$, the integral of the meromorphic function $f(p)$ over the set $A\cap \tilde \C$; 

\item [-] we denote by $\Gamma$ the complex path such that $E(p)$ monotonically maps $\Gamma$ onto $\R$ (see Figure \ref {Fig1}). \ We 
denote also by
\bee
\Gamma_1 = \left [- \frac 12 N b , +\frac 12 N b \right ] \cap \tilde \R \ \ \mbox { and } \ \Gamma_c = \tilde \R \setminus \Gamma_1 
\eee
where $N$ is fixed and such that $\kappa_N >0$; by definition $\fint_{\Gamma_1} f(p) dp$ coincides 
with $\int_{\Gamma_1} f(p) dp$ and, similarly, $\fint_{\Gamma_c} f(p) dp$ coincides with $\int_{\Gamma_c} f(p) dp$;

\item [-] We define ${\mathcal H}_E$ the Hilbert space of functions $\hat \psi (p) \in L^2 (\tilde \R , dp)$ such that 
\bee
\frac {d^r \hat \psi \left ( - \frac 12 Nb \pm 0 \right )}{dp^r} = \frac {d^r \hat \psi \left (  \frac 12 Nb \mp 0 \right )}{dp^r}, \ 
\forall r \in \N \, . 
\eee
Hence, we may identify the vector $\hat \psi$ with the vector $(\hat \psi_1 , \hat \psi_c ) \in L^2 (\Gamma_1 ) 
\otimes L^2 (\Gamma_c )$ satisfying the periodic boundary conditions at the points $\pm \frac 12 N b$, that is
\bee
\psi_1 \left ( - \frac 12 N b +0 \right ) = \psi_1 \left ( + \frac 12 N b -0 \right ) 
\eee
and
\bee 
\psi_c \left ( - \frac 12 N b -0 \right ) = \psi_c \left ( + \frac 12 N b +0 \right )
\eee
together with their derivatives. \ In particular, in the case $N=1$ then $L^2 (\Gamma_1 )$ coincides with the space $L^2 ({\mathcal B})$.

\end {itemize}

\section {Main Result} 

Here we state the assumptions on the model and the main result.

\begin {theorem} \label {Theorem1}
Let $V(x)$ be a non constant periodic function with period $a$
\bee
V(x) = V(x+ a) \, , \ \forall x \in \R \, , 
\eee
and mean value zero:
\be
\frac 1a \int_0^a V(x) dx =0 \, . \label {Eqn0}
\ee
We assume that $V(x)$ is an analytic function in an open set containing the strip $|\Im x| \le R$ for some $R>0$. \ Let 
$N\ge 1$ any fixed integer number such that the $N$-th gap of the Bloch operator $H_B$ is not empty. \ Then there exists $F_N >0$ such 
that for any $F\in (0,F_N)$ the Stark-Wannier operator $H_F$ admits $N$ ladders of quantum resonances; that is the resolvent 
operator $[H_F -z]^{-1}$ has an analytic continuation from the upper half-plane $\Im z >0$ to the lower half-plane $\Im z <0$, with 
poles corresponding to the quantum resonances.
\end {theorem}

\begin {remark}
In fact, the assumption (\ref {Eqn0}) on the mean value of $V(x)$ is for sake of definiteness. \ Indeed, if the mean value of the potential 
is not zero then we can always recover condition (\ref {Eqn0}) by means of the phase gauge $\psi \to \psi e^{-i t \int_0^a V(x) dx /a \hbar }$. 
\end {remark}

\begin {remark}
The regular perturbation theory we apply in \S 7 gives that the imaginary part of the quantum resonances, which is strictly negative, is 
of order $O(F^2)$ in the limit of small $F$. \ In fact, we expect that the imaginary part is exponentially small in the limit of small $F$ as 
it has been proved in the case of periodic potential with a finite number of gaps \cite {GS}; however, we don't dwell here on such a problem.
\end {remark}

\section {Preliminary results} \label {Sec3}

Here we consider the spectral problem for Stark-Wannier operator
\bee
H_F := - \frac {d^2}{dx^2} + V(x) + F x \, , 
\eee
where $F >0$ is a given parameter and where $V(x)$ is a ({\it not constant}) smooth periodic potential with period $a$:
\bee
V(x) = V(x+ a) 
\eee
and mean value zero.

We recall now some basic properties of the Bloch and band functions \cite {RS}.

\subsection {Bloch Decomposition}

Let $H_B$ be the Bloch operator formally defined on $L^2 (\R , dx)$ as
\bee
H_B := - \frac {d^2}{dx^2} + V(x) \, .
\eee
Its spectrum is given by bands. \ Let $k \in {\mathcal B}$ be the quasi-momentum (or crystal momentum) variable, the torus 
${\mathcal B} = \R / b\Z = \left ( - \frac 12 b, + \frac 12 b \right ]$, where $b= \frac {2\pi}{a}$, is usually named Brillouin zone.

Let $\varphi_n (x,k)$ denote the Bloch functions associated to the band functions $E_n (k)$; where the Bloch functions are assumed to be 
normalized to $1$ on the interval $[0,a]$: 
\be
\int_0^a \overline {\varphi_m (x,k)} \varphi_n (x,k)  d x =\delta_m^n \, . \label {Eqn1}
\ee
The band and Bloch functions satisfy to the following eigenvalues problem 
\be
H_B \varphi = E \varphi \label {Eqn1Bis}
\ee
with quasi-periodic boundary conditions (hereafter we denote $' = \frac {\partial }{\partial x}$, as usual)
\bee
\varphi (a,k) =e^{i ka} \varphi (0,a) \ \mbox { and } \ \varphi' (a,k) =e^{i ka} \varphi' (0,k) \, . 
\eee
The Bloch functions $\varphi_n$ may be written as
\bee
\varphi_n (x,k) = e^{i k x} u_n (x,k) 
\eee
where $u_n (x,k)$ is a periodic function with respect to $x$: $u_n (x+a , k) = u_n (x,k)$. \ For any fixed $k \in {\mathcal B}$ the 
spectral problem (\ref {Eqn1Bis}) has a sequence of real eigenvalues
\bee
E_1 (k) \le E_2 (k) \le \cdots \le E_n (k) \le \cdots \, ,
\eee
such that $\lim_{n\to \infty} E_n (k) = + \infty$. 

As functions on $k$, both Bloch and band functions are periodic with respect to $k$:
\bee
E_n (k) =E_n (k+b) \ \mbox { and } \ \varphi_n (x,k)=\varphi_n (x,k+b) \, ,
\eee
and they satisfy to the following properties for any real-valued $k$:
\bee
\varphi_n (x,-k)= \overline {\varphi_n (x,k)} \ \ \mbox { and } \ \ {E}_n (-k) = {E}_n (k) \, . 
\eee
Furthermore, if $V(x)$ is an even potential then $\varphi_n (-x,k)= \overline {\varphi_n (x,k)}$, $\varphi_n (x,0)$ are even functions 
while $\varphi_n (x, b/2)$ are odd functions. 

The band functions $E_n (k)$ are monotone increasing (resp. decreasing) functions for any $k \in \left [ 0 , \frac 12 b \right ]$ if 
the index $n$ is an odd (resp. even) natural number. \ The spectrum of $H_B$ is purely absolutely continuous and it is given by bands:
\bee
\sigma (H_B) = \cup_{n=1}^\infty [E^b_n , E^t_n] \ \ \mbox { where } \ \ [E^b_n , E^t_n] = \{ E_n (k) ,\ k \in {\mathcal B} \} \, . 
\eee
In particular we have that 
\bee
E^b_n = 
\left \{
\begin {array}{ll}
E_n (0) & \ \mbox { for odd } n \\ 
E_n (b/2) & \ \mbox { for even } n
\end {array}
\right. \ \mbox { and } \ 
E^t_n = 
\left \{
\begin {array}{ll}
E_n (b/2) & \ \mbox { for odd } n \\ 
E_n (0) & \ \mbox { for even } n
\end {array}
\right. \, . 
\eee
The intervals $(E^t_n, E^b_{n+1})$ are named gaps; a gap $(E^t_n, E^b_{n+1})$ may be empty, that is $E^b_{n+1}=E^t_n$, or not. \ It is 
well known that, in the case of one-dimensional crystals, all the gaps are empty if, and only if, the periodic potential is a constant 
function. \ Because we assume that the periodic potential is not a constant function then one gap, at least, is not empty.

From the Bloch decomposition formula it follows that any vector $\psi \in L^2$ can be written as \cite {OK}
\bee
\psi (x) = \sum_{n \in \N}  \frac {1}{2\pi} \int_{{\mathcal B}} \varphi_n (x,k ) \hat \psi_n (k ) d k  \, . 
\eee
The family of functions $\{ \hat \psi_n (k ) \}_n$ is called the crystal momentum 
representation (hereafter CMR) of the wave vector associated to $\psi$  and it is defined as
\bee
\hat \psi_n (k ) = \int_{\R} \overline {\varphi_n (x, k )} \psi (x)  d x \, . 
\eee
By construction any function $\hat \psi_n (k)$ is a periodic function and the transformation 
\be
\psi \in L^2 (\R , dx) \to {U} \psi = \hat \psi := \{ \hat \psi_n \}_{n=1}^{\infty} \in {\mathcal H} := \otimes_{n=1}^\infty 
L^2 ({\mathcal B} , d k ) \label {Eqn1Ter}
\ee
is unitary:
\bee
\| \psi \|_{L^2 (\R , dx )}^2 = \left \| \hat \psi \right \|_{\mathcal H}^2 :=  \sum_{n=1}^\infty \| \hat \psi_n \|_{L^2 ({\mathcal B}, 
dk )}^2\, .
\eee

\subsection {Analytic properties of the band and Bloch functions} \label {Sec3.2} The band  functions are analytic functions with 
respect to $k \in \C$; more precisely, they are the branches of a single multisheeted function that has no singularities other that the square 
root branch points (hereafter named Kohn branch points) $k_n = \frac 12 n b \pm i \kappa_n$, $n=\pm 1 , \pm 2 , \ldots , $, for 
some $\kappa_n =\kappa_{-n} \ge 0$ \cite {K}. \ In fact, the Kohn branch points only occur in correspondence of open gaps. \ If the 
$n-$th gap is open (that is $E^b_{n+1}>E^t_n$) then $\kappa_n >0$. \ On the other side, if the $n-$th gap is empty (that is 
$E^b_{n+1}=E^t_n$)then the Kohn branch points $\pm k_n$ don't occur; actually, in such a case we have that $\kappa_n =0$ and the two 
branch points $k_n$ and $\bar k_n$ coincide and then the multisheeted band and Bloch functions are regular at this point.

Let $\ell_n$ be the straight lines connecting the branch points $\kappa_n$ and $\bar \kappa_n$, let $\tilde \C =\C -\cup_{n=\pm 1,\pm 2, 
\ldots } \ell_n$. \ Then there exists an analytic multisheeted function ${E}(p)$ (hereafter named \emph {multisheeted band function}) such 
that 
\be
E_n (k) = 
\left \{
\begin {array}{ll}
{E} \left [ \frac 12 (n-1) b + k \right ] &   \ \mbox { if } \ n \ \mbox { is odd} \\ 
   &   \\ 
{E} \left [ \frac 12 n b - k \right ]  & \ \mbox { if } \ n \ \mbox { is even}
\end {array}
\right. \, ,  k \in \left [ 0, \frac 12 b \right ] \,  . 
\ee
Furthermore, the multisheeted band function is such that $E (\bar p)= \overline E (p)$.

Let us denote by $\Gamma$ the complex path defined as in Figure \ref {Fig1}, then the multisheeted band function ${E}(p)$ monotically maps 
$\Gamma$ onto $\R$ \cite {F}. \ In particular ${E}(p)$ maps the real intervals $\left [ \frac {n-1}{2} b , \frac n2 b \right ]$ onto the 
interval band $[E^b_n , E^t_n ]$, and it maps the clockwise path surrounding the cut $\ell_n$ associated to the Kohn branch point $k_n$ 
to the gap $(E^t_n , E^b_{n+1})$.

Similarly, the Bloch functions $\varphi_n (x,k)$ are the branches of a single multisheeted Bloch function $\varphi (x,p)$ (hereafter 
named \emph {multisheeted Bloch function}) that has no singularities other than the Kohn branch points. \ As before, we may write 
\bee
\varphi (x,p) = e^{i p x} u (x,p)
\eee
where $u(x,p)=u(x+a,p)$ is a periodic function.

Then the multisheeted band function and the multisheeted Bloch function are defined for any $p \in \tilde \C$ and, in particular, on 
the set $\tilde \R = \tilde \C \cap \R$, are such that
\be
E\left ( - \frac 12 n b \pm 0 \right ) = E\left ( \frac 12 n b \mp 0 \right ) \ \mbox { and } \ \varphi \left ( x, - \frac 12 n b \pm 0 
\right ) = \varphi \left ( x, \frac 12 n b \mp 0 \right ) \label {Eqn2Bis}
\ee
together with their derivatives, for any $n$ such that the $n^{th}$ gap is open. \ In particular, we have that
\bee
E\left ( - \frac 12 n b + 0 \right ) = E\left ( \frac 12 n b - 0 \right ) = E^t_n 
\eee
and 
\bee 
E\left ( - \frac 12 n b - 0 \right ) = E\left ( \frac 12 n b + 0 \right ) = E^b_{n+1} \, . 
\eee

\begin{figure} [ht]
\begin{center}
\includegraphics[height=6cm,width=12cm]{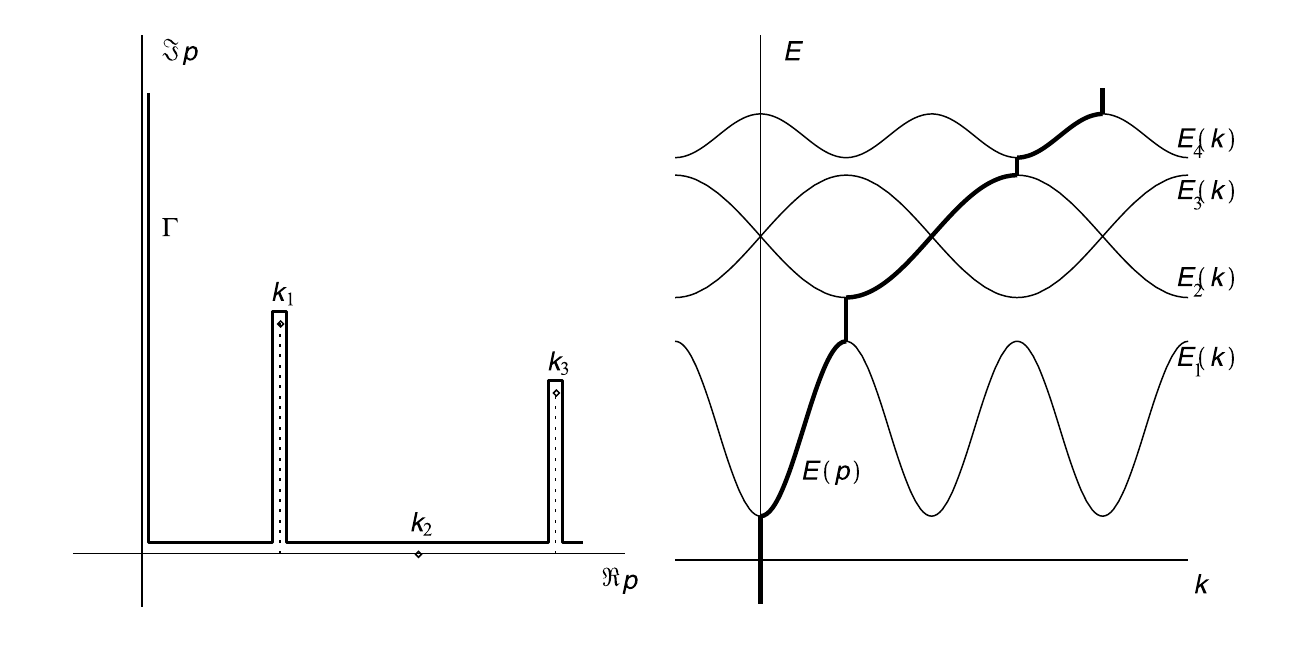}
\end{center}
\caption{\label {Fig1} In the left panel we plot the path $\Gamma$ in the quasimomentum complex plane where $k_1$, $k_2$ and $k_3$ 
are the first 3 Kohn's branch points. \ In the right panel we plot the first 4 band functions and the multisheeted function $E(p)$ 
which maps $\Gamma$ onto the whole real axis (bold line). \ When $p$ is real valued then $E(p)$ takes values in the spectrum of $H_B$; 
when $p$ belongs to the imaginary axes then the energy $E(p)$ lies below the spectrum of $H_B$; and finally, when $p$ belongs the 
part of $\Gamma$ surrounding the complex cuts $\ell_n$ then the energy $E(p)$ belongs to the gaps. \ In this picture we have the situation 
where the gap between the $2^{nd}$ and $3^{rd}$ band is empty, this means that the two square root Kohn's branch points $k_2$ and 
$\bar k_2$ coincides and then the multisheeted function $E(p)$ is regular at this points, in fact. \ We should remark that the 
\emph {periodic} band function $E_2 (k)$ and $E_3 (k)$, and the associated Bloch functions, are \emph {singular} at the point $k=0$; 
this ''artificial'' singularity is a consequence of the periodicity of the band functions and from the fact that the gap between the 
$2^{nd}$ and the $3^{rd}$ band is empty.}
\end{figure}

\section {Asymptotic behavior of the multisheeted Bloch and band functions} \label {Sec4}

We consider the differential equation
\be
- \varphi'' (x) + V(x) \varphi (x) = E \varphi (x) \label {Eqn3}
\ee
where $V(x)=V(x+a)$ is a smooth periodic function. \ We denote by $\phi_1 (x,E)$ and $\phi_2(x,E)$ two independent solutions of this 
equation with initial conditions
\be
\phi_1 (0,E)=1 , \ \phi_1 ' (0,E)=0 \ \mbox { and } \ \phi_2 (0,E)=0 , \ \phi_2 ' (0,E)=1 \, .  \label {F7}
\ee 

The general solution of equation (\ref {Eqn3}) may be written as 
\be
\varphi := \varphi (x,E) = \alpha \phi_1 (x,E) + \beta \phi_2 (x,E) \label {Eqn5}
\ee
where the parameters $\alpha$ and $\beta$ are chosen in order to normalize to one the solution $\varphi$:
\be
\int_0^a |\varphi (x,E)|^2 dx =1\, . \label {Eqn6}
\ee

Let $p \in \R$ and let 
\bee
\lambda = e^{i p a }\, . 
\eee
If we look for solutions satisfying the boundary conditions
\be
\varphi (a,E) = \lambda \varphi (a,E) \ \mbox { and } \ \varphi' (a,E) = \lambda \varphi' (a,E)   \label {F10}
\ee
then $\lambda$ should satisfies the equation
\bee
\lambda^2 - 2 \mu (E) \lambda + 1 =0
\eee
where $\mu (E)= \frac 12 \left [ \phi_1 (a,E) + \phi_2' (a,E) \right ]$. \ In conclusion, $p$ may be written as function of $E$ 
(and viceversa) by means of the relation
\bee 
\cos (p a) = \mu (E) \, , \ \mbox { that is } \ \frac 12 \left ( \lambda + \frac {1}{\lambda} \right ) = \mu (E) \, . 
\eee 
The solution $E=E(p)$ of the equation above is the multisheeted band function, and the associated wave-function  
$\varphi = \varphi (x,p) = \varphi [x,E(p)]$ is the multisheeted Bloch function.

\begin {remark} 
In the case of symmetric potentials $V(x)=V(-x)$ then it follows that $ \phi_1 (a,E) = \phi_2' (a,E) $ (see, e.g., \S 1.3 \cite {MW}) 
and we have that
\bee
\mu (E) = \phi_1 (a,E)\, .
\eee
\end {remark}

We prove now that

\begin {lemma} \label {Lemma1Bis}  
Let $\phi_1 (x,E)$ and $\phi_2 (x,E)$ be the solutions of Eq. (\ref {Eqn3}) satisfying the initial condition (\ref {F7}). \ Let $\varphi (x,E)$ be the normalized solution of Eq. (\ref {Eqn3}) satisfying the boundary conditions (\ref {F10}). \ Then
\be
\varphi (x,E) = C(E)  \left [ \phi_1 (x,E) + i \sqrt E \phi_2 (x,E) \right ] \, , \ C(E) = \left [ \sqrt {\frac {1}{a}} + O(E^{-3/2} ) \right ] \, , \label {Eqn12Bis} 
\ee
for large $E$. \ The above asymptotic behavior can be derived with respect to $E$ term by term.
\end {lemma}

\begin {proof}
Hereafter, we drop out the dependence on $E$ when this does not cause misunderstanding. \ As a first step we prove that the equation (\ref {Eqn3}) has non-trivial solutions with $\phi_2 (a)=0$ (or $\phi_1'(a)=0$) only for values of $E$ such that $\lambda = \pm 1$. \ Indeed, the quasiperiodic boundary conditions implies that 
\bee
\left \{
\begin {array}{lcl}
\left [ \phi_1 (a) - \lambda \right ] \alpha + \phi_2 (a) \beta &=& 0 \\ 
\phi_1' (a) \alpha + \left [ \phi_2' (a) - \lambda \right ] \beta &=& 0 
\end {array}
\right.
\eee
that is
\bee
\left \{
\begin {array}{lcl}
\frac 12 \left [ \frac {1}{\lambda} - \lambda \right ] \alpha + \phi_2 (a) \beta &=& 0 \\ 
\phi_1' (a) \alpha + \frac 12 \left [ \frac {1}{\lambda} - \lambda \right ] \beta &=& 0 
\end {array}
\right.
\eee
If we assume that $\lambda \not= \pm 1$ then $\phi_2 (a) \not=0$ and $\phi_1'(a) \not= 0$. \ Indeed, if, for instance, $\phi_2 (a)=0$ 
then $\alpha =0$ and finally even $\beta $ should be equal to zero, that is we only have a trivial solution. 

Thus, assuming that $\lambda \not= \pm 1$ the wave function (\ref {Eqn5}) may be written in the form \cite {K} (adapted to the normalization 
condition (\ref {Eqn6}))
\be
\varphi (x) = \frac {\phi_2 (a) \phi_1 (x) + \frac 12  \left ( \frac {1}{\lambda} - \lambda \right ) \phi_2 (x)}
{\sqrt {-2 \phi_2 (a) \frac {d \mu}{dE}}} = \frac {\phi_2 (a) \phi_1 (x) + i \sin (pa)  \phi_2 (x)}{\sqrt {-2 \phi_2 (a) 
\frac {d \mu}{dE}}} \label {Eqn11} 
\ee
where we observe that $ \frac 12  \left ( \frac {1}{\lambda} - \lambda \right ) = i \sin (pa)$. \ At $\lambda = \pm 1$ 
the solution $\varphi (x)$ follows from (\ref {Eqn11}) by means of a continuity argument.

Recalling that the non trivial solutions $\phi_{1,2}$ satisfy the following integral equation (see Lemma 1.7 \cite {T}):
\be
\phi_1 (x,E) &=& \cos (\sqrt E x ) + \frac {1}{\sqrt E} \int_0^x \sin \left [ \sqrt E (x-y) \right ] V(y) 
\phi_1 (y,E) dy  \label {Eqn12} \\
\phi_2 (x,E) &=& \frac {1}{\sqrt E} \sin (\sqrt E x ) + \frac {1}{\sqrt E} \int_0^x \sin \left [ \sqrt E (x-y) \right ] 
V(y) \phi_2 (y,E) dy \label {Eqn13}
\ee
from which follows the asymptotic behavior for large $E$: 
\bee
\phi_1 (x,E) &=& \cos (\sqrt E x) + O \left ( \frac {1}{\sqrt E} \right ) \, , 
\\
\phi_2 (x,E) &=& \frac {\sin (\sqrt E x)}{\sqrt E} + O \left ( \frac {1}{E} \right ) \, ,  
\eee
where the asymptotic behavior is uniform for any $x \in [0,a]$ and where the asymptotic behavior can be derived with respect to $x$ and 
$E$ term by term; e.g.:
\bee
\frac {\partial \phi_1 (x,E) }{\partial x} &=& - \sqrt E \sin (\sqrt E x) + 
O \left (1 \right )\, ,  
\\
\frac {\partial \phi_1 (x,E) }{\partial \sqrt E} &=& - x \sin (\sqrt E x) +  O \left ( \frac {1}{\sqrt E} \right ) \, ,  
\eee
and
\bee
\frac {\partial \phi_2 (x,E) }{\partial x} &=&\cos (\sqrt E x) + O \left ( \frac {1}{\sqrt E} \right ) \, , 
\\
\frac {\partial \phi_2 (x,E) }{\partial \sqrt E} &=&\frac {x\cos (\sqrt E x)}{\sqrt E}  +  O \left ( \frac {1}{E} 
\right ) \, . 
\eee

In order to have an asymptotic expression with more terms we iterate equations (\ref {Eqn12}) and (\ref {Eqn13}) obtaining the following  
asymptotic behavior for large $E$, which are uniform for any $x \in [0,a]$
\bee
\phi_1 (x) 
&=& \cos (\sqrt E x ) + \frac {1}{2\sqrt E} \sin(\sqrt E x) Q(x) + \frac {1}{4E} \left [ V(x) - V(0) \right ] \cos \left ( \sqrt {E} x 
\right ) + g_1 (x,E) \\ 
\phi_2 (x) 
&=& \frac {\sin (\sqrt E x )}{\sqrt E} - \frac {1}{2 E} \cos(\sqrt E x) Q(x) + \frac {1}{4E^{3/2}} \left [ V(x) - V(0) \right ] 
\sin \left ( \sqrt {E} x \right ) + g_2 (x,E)
\eee
where $Q(x) = \int_0^x V(y) dy $ and where 
\bee
g_1 (x,{E}) &=& - \frac {1}{4E} \int_0^x \cos \left [ \sqrt {E} (x-2y) \right ] V' (y) dy + \frac {1}{2E} \int_0^x V(z) \phi_1 (z,E) 
dz \times \\ 
&& \ \times \left [ \int_z^x \cos [ \sqrt {E} (x+z-2y) ] V(y) dy - \cos [\sqrt {E} (x-z) ] [Q(z)-Q(x) ] \right ] \\
g_2 (x,{E}) &=& - \frac {1}{4E^{3/2}} \int_0^x \sin \left [ \sqrt {E} (x-2y) \right ] V' (y) dy + \frac {1}{2E^{3/2}} \int_0^x V(z) 
\phi_2 (z,E) dz \times \\ 
&& \ \times \left [ \int_z^x \cos [ \sqrt {E} (x+z-2y) ] V(y) dy - \cos [\sqrt {E} (x-z) ] [Q(z)-Q(x) ] \right ] 
\eee
are such that
\bee
g_1 (x,E) \sim E^{-3/2} , \ \frac {\partial g_1 (x,E)}{\partial x} \sim E^{-1} \ \mbox { and } \ \frac {\partial g_1 (x,E)}{\partial E} 
\sim E^{-3/2} \\ 
g_2 (x,E) \sim E^{-2} , \ \frac {\partial g_2 (x,E)}{\partial x} \sim E^{-3/2} \ \mbox { and } \ \frac {\partial g_2 (x,E)}{\partial E} 
\sim E^{-2}
\eee
by integrating by part. \ In particular, recalling that $Q(a) =\int_0^a V(x) dx=0$, we have that
\bee
\mu (E) &=& \cos (\sqrt E a ) + \frac {\sin (\sqrt E a )}{2\sqrt E} \int_0^a V(y) dy + O(E^{-1}) \\ 
&=&  \cos (\sqrt E a )  + O(E^{-1}) \\
\phi_2 (a) &=& \frac {1}{\sqrt {E}} \sin (\sqrt E a ) + O(E^{-3/2}) \\ 
\frac {d \mu }{dE} &=& - \frac {a}{2\sqrt E} \sin (\sqrt E a) + O(E^{-2}) \\ 
\eee
and thus $p = \sqrt E + O (E^{-1}) $, $\sin (pa) = \sin (\sqrt E a ) + O(E^{-1}) $ and 
\bee
\varphi (x,E) &=& \sqrt {- \frac {\phi_2 (a,E)}{2 \frac {d\mu}{dE}}} \phi_1 (x,E) + 
\frac {\sin (pa)}{\sqrt {- 2 \phi_2 (a,E) \frac {d\mu}{dE}}} \phi_2 (x,E) \\ 
&=& \left [ \sqrt {\frac {1}{a}} + O(E^{-3/2} ) \right ] \left [ \phi_1 (x,E) + i \sqrt E \phi_2 (x,E) \right ] 
\eee
%&=& \left [ \sqrt {\frac {1}{a}} + O(E^{-3/2} ) \right ]  \left [ e^{i \sqrt E x} - 
%\frac {i}{2\sqrt E} e^{i\sqrt E x} Q(x) +  e^{i \sqrt E x} \frac {V(x)-V(0)}{4E} + O(E^{-1}) \right ] \nonumber \\ 
%&=&\sqrt {\frac {1}{a}} e^{i \sqrt E x} \left [ 1 +  \frac {V(x)-V(0)}{4E} \right ] - \frac {i}{2\sqrt {Ea}} e^{i\sqrt E x} Q(x) + O(E^{-1})  \nonumber 
%
proving the asymptotic behavior (\ref {Eqn12Bis}). 

We have only to discuss the case of non symmetric periodic potentials. \ In the case of non symmetric potential the expression ot the 
wavefunction $\varphi $ given by (\ref {Eqn11}) must be replaced by (see Appendix \ref {AppA})
\be
\varphi (x) =  \frac {\phi_2 (a) \phi_1 (x) - \left ( \phi_1 (a) - \lambda \right ) \phi_2 (x)}{\sqrt {-2 \phi_2 (a) \frac {d \mu}{dE}}}
\label {Eqn20}
\ee
and even in such a case the asymptotic arguments as above may be applied. \ The Lemma \ref {Lemma1Bis} is thus completely proved. 
\end {proof}

Now, we are ready to prove the estimate of $u(x,E)$ for any complex value $x$ belonging to the box
\bee
L = \left \{ x \in \C \ : \ \Re x \in [0,a], \ \Im x \in [-R,0] \right \}.
\eee
In fact, the solution of equation (\ref {Eqn3}) may be extended to complex values $x$ in the strip $|\Im x |\le R$ because the potential $V$ is an analityc function on such a strip.

\begin {theorem} \label {Lemma1}
The multisheeted Bloch function $\varphi (x,E) = e^{ipx} u(x,E)$ is such that 
\be
u (x,E) = \frac {1}{\sqrt a}  -i \frac {Q(x)}{2\sqrt {a E}} + O \left ( E^{-1} \right ) \, , \ \mbox { where } \ 
Q(x) =\int_{\gamma_{0,x}} V(y) dy \, , \label {F16}
\ee 
as $E \to \infty$ uniformly with respect to $x \in L$, and ${\gamma_{0,x}}$ is any complex path contained in the box $L$ and connecting the complex points $0$ and $x$. \ The asymptotic behavior can be derived with respect to $x$ and $E$ term by term; in particular, 
\bee
\frac {\partial u(x,E)}{\partial E} = i \frac {Q(x)}{4 E^{3/2} \sqrt a} + O (E^{-2}) 
\eee
and
\bee
\frac {\partial^2 u(x,E)}{\partial E^2} = -i \frac 38 \frac {Q(x)}{ E^{5/2} \sqrt a} + O (E^{-3}) \, . 
\eee
\end {theorem}

\begin {proof}
From (\ref {Eqn12Bis}) it follows that
\bee
\varphi (x,E) = C(E) \phi (x,E) ,\ \phi (x,E):= \phi_1 (x,E) + i \sqrt {E} \phi_2 (x,E)
\eee
where, recalling (\ref {Eqn12}) and (\ref {Eqn13}), it follows that 
\bee
\phi (x,E) = e^{i\sqrt {E} x} + \frac {1}{\sqrt {E}} \int_{\gamma_{0,x}} \sin \left [ \sqrt {E} (x-y) \right ] V(y) \phi (y,E) d y \, . 
\eee
If we choose the path ${\gamma_{0,x}}$ such that $\Im x \le \Im y$ for any $y \in {\gamma_{0,x}}$, and if we set
\bee
\phi (x,E) = e^{i\sqrt E x} v(x,E) 
\eee
then $v(x,E)$ satisfyies to the following integral equation 
\bee
v(x,E) 
&=& 1 + \frac {1}{\sqrt E} \int_{\gamma_{0,x}} \sin \left [ \sqrt {E} (x-y) \right ] e^{-i \sqrt E (x-y)} V(y) v (y,E) d y \\ 
&=& 1 + \frac {1}{\sqrt E} \int_{\gamma_{0,x}} \sin \left [ \sqrt {E} (x-y) \right ] e^{-i \sqrt E (x-y)} V(y) d y + r(x,E)
\eee
where
\bee
r(x,E) &=& \frac {1}{E} \int_{\gamma_{0,x}} \sin \left [ \sqrt {E} (x-y) \right ] e^{-i \sqrt E (x-y)} V(y) \times \\ && \ \ \times \int_{\gamma_{0,y}} \sin \left [ \sqrt {E} (y-z) \right ] e^{-i \sqrt E (y-z)} V(z) v(z,E) dz d y
\eee
and
\bee
\sin \left [ \sqrt {E} (x-y) \right ] e^{-i \sqrt E (x-y)} = \frac {1-e^{-2i \sqrt E (x-y)}}{2i}
\eee
is such that 
\bee
\left | e^{-2i \sqrt E (x-y)} \right | = e^{2 \sqrt E (\Im x- \Im y)} \le 1 \, , \ \mbox { for any } y\in {\gamma_{0,x}} \, .
\eee
Hence
\bee
v(x,E)=1 -i \frac {Q(x)}{2\sqrt E} + \frac {i}{2\sqrt E} \int_{\gamma_{0,x}} e^{-2i \sqrt E (x-y)} V(y) dy + r(x,E)
\eee
If we set 
\bee
\nu (E) = \sup_{x\in L } |v(x,E)|
\eee
then from the equation above it turns out that
\bee
\nu \le 1 + \frac {C_1}{\sqrt E} + \frac {C_2}{E} \nu
\eee
for some positive constants $C_1$ and $C_2$. \ From this fact, it follows that $\nu \le C$ for some positive constant $C$ uniformly with respect to $E$. \ Hence we can conclude that 
\bee
v (x,E) = 1 - i \frac {Q(x)}{2\sqrt E} + \asy (E^{-1}) 
\eee
in the limit of large $E$. \ Finally
\bee
u(x,E) = e^{-i p x} \varphi (x,E) = e^{-i px} C(E) \phi (x,E) = C(E) e^{-i p x} e^{i \sqrt E x} v (x,E) 
\eee
from which (\ref {F16}) follows for any $x\in L$. \ We close by underlining that, by construction, the asymptotic behavior can be derived term by term.
\end {proof}

\section {Stark-Wannier operator in the extended crystal momentum representation}

The CMR of the Stark-Wannier operator is usually defined on ${\mathcal H} = \otimes_{n=1}^{+\infty} L^2 ({\mathcal B},dk)$ as 
(see \cite {B})
\bee
\left [ U H_F U^{-1} \hat \psi \right ]_m (k) = i F \frac {\partial \hat \psi_m (k)}{\partial k} + E_m (k) \hat \psi_m (k) + 
F \sum_{n\in \N} X_{m,n} (k) \hat \psi_n (k) \, ,
\eee
where $U$ is defined by (\ref {Eqn1Ter}) and the interband coupling terms are defined as 
\be
X_{m,n} (k) = \frac {i}{a} \int_{0}^{a} \overline {u_m (x,k)} \frac {\partial u_n (x,k)}{\partial k} d x \, . \label {Eqn22}
\ee

This usual approach is not suitable in order to apply analytic distortion techniques; hence, as a preliminary step, by making use of 
the analytic properties of the multisheeted band and Bloch functions $E(p)$ and $\varphi (x,p)$, we  
define the Bloch decomposition formula and the CMR  in an equivalent form which will be useful in the sequel 
(it will be called \emph {extended} crystal momentum representation, or simply ECMR). \ 

The transformation $U$ can be rewritten as $\hat \psi = U_E \psi $ where  
\bee
\hat \psi (p) = \int_{\R} \overline {\varphi (x,p)} \psi (x) d x \,  , p \in \tilde \R \, . 
\eee
The inverse transformation is defined as  
\bee
\left ( U_E^{-1} \hat \psi \right )(x)= \frac {1}{2\pi} \int_{\tilde \R} \varphi (x,p) \hat \psi (p) dp  \, . 
\eee

\begin {lemma} \label {Lemma3}
Let $\psi (x)$ be a rapidly decreasing function: that is for any $\alpha , \beta \in \N$ there exists a positive constant 
$C=C_{\alpha , \beta}$ such that 
\bee
\left | x^\alpha \frac {d^r \psi (x)}{d x^r} \right | \le C \, , \ \forall x \in \R \, . 
\eee
Then $\hat \psi = U_E \psi$ is a rapidly decreasing function, too. \ And vice versa.
\end {lemma}

\begin {proof}
Let 
\bee
\hat \psi (p) = \int_{\R} e^{-i p x} \overline {u(x,p)} \psi (x) d x \, . 
\eee
Then, by integrating by parts 
\bee
p^\alpha \hat \psi (p) 
= \int_{\R} p^\alpha e^{-ipx} \overline {u(x,p)}  \psi (x) dx  
= (-i)^\alpha \int_{\R} e^{-ipx} \frac {\partial^\alpha}{\partial x^\alpha} \left [ \overline {u(x,p)} \psi (x) \right ] dx 
\eee
is bounded because $\overline {u(x,p)}$ is  s smooth periodic function with respect to $x$ and $\psi (x)$ is a rapidly decreasing 
function. \ For the same reason it follows that 
\bee
\frac {d^\beta \hat \psi (p)}{dp^\beta} = \int_{\R} \frac {\partial^\beta \left ( e^{-ipx} \overline {u(x,p)}\right )}{\partial p^\beta} 
\psi (x) dx 
\eee
converges. \ Finally, the estimates are uniformly with respect to $p$ because of Theorem \ref {Lemma1}. \ Similarly, the vice versa follows.
\end {proof}

\begin {theorem}\label {Theorem2}
If we define 
\bee
\tilde H_F = U_E H_F U_E^{-1} 
\eee
then it acts as
\be
\left [ \tilde H_F \hat \psi \right ] (p) &=& E (p) \hat \psi (p) + i F \frac {\partial \hat \psi (p)}{\partial p} +  
F \left ( X \hat \psi \right ) (p ) \, , \ p \in {\tilde \R} \, , \label {Eqn23} 
\ee
where $X $ represents the interband coupling term:
\bee
\left ( X \hat \psi \right ) (p) &=& \sum_{j \in \Z} C_j (p) \hat \psi (p-jb)  \, . 
\eee
where
\bee
C_j (p) = \frac ia \int_0^a e^{-i b j x} \bar u (x,p) \frac {\partial u (x,p-bj)}{\partial p} d x \, . 
\eee
\end {theorem}

\begin {proof}
Indeed, let $\hat \psi$  be a test function (e.g. $\hat \psi$ is a rapidly decreasing function)
\bee
\left [ H_F U_E^{-1} \hat \psi \right ] (x) = \frac {1}{2\pi} \int_{\tilde \R} H_F \varphi (x,p) \hat \psi (p) dp 
\eee
where 
\bee
H_F \varphi (x,p) = H_B \varphi (x,p) + Fx \varphi (x,p) = E(p) \varphi (x,p) + F x \varphi (x,p) 
\eee
and where 
\bee
&& \int_{\tilde \R} x \varphi (x,p) \hat \psi (p) dp = \int_{\tilde \R} x e^{ipx} u (x,p) \hat \psi (p) dp 
= -i \int_{\tilde \R} \frac {\partial e^{ipx}}{\partial p} u (x,p) \hat \psi (p) dp \\ 
&& \ = -i \left [\varphi (x,p) \hat \psi (p) \right ]_{\partial \tilde \R} + i \int_{\tilde \R} e^{ipx} \frac {\partial u (x,p) }
{\partial p} \hat \psi (p) dp + i \int_{\tilde \R} \varphi (x,p) \frac {\partial \hat \psi (p) }{\partial p} dp 
\eee
because $\varphi (x,p)$ and $\hat \psi (p)$ have the same values, together with their derivatives, at the boundary points of 
$\tilde \R$ (see Eq. (\ref {Eqn2Bis})). \ Hence (\ref {Eqn23}) follows where 
\be
\left ( X \hat \psi \right ) (p)  &=& \frac {i}{2\pi} \int_{\R} \overline {\varphi (x,p)} \int_{\tilde \R} e^{ip' x}  
\frac {\partial u (x,p')}{\partial p} \hat \psi (p') d p' d x \, . \label {Eqn24}
\ee 
In particular, a formal calculation lead us to the following result
\be
\left ( X \hat \psi \right )(p) &=& \frac {i}{2\pi} \int_{\tilde \R} \left [ \int_{\R} \overline {\varphi (x,p)}  e^{ip' x} 
\frac {\partial u (x,p')}{\partial p}  d x \right ] \hat \psi (p') d p' \label {Eqn25}\\ 
& =& \frac {i}{2\pi} \int_{\tilde \R} \left [ \int_{\R} \overline {u (x,p)}  e^{i(p'-p) x} \frac {\partial u (x,p')}{\partial p}  
d x \right ] \hat \psi (p') d p' \nonumber \\ 
&=& \frac {i}{2\pi} \int_{\tilde \R} \left [ \int_{\R} \sum_{j\in \Z} \alpha_j (p,p') e^{i \left ( p'-p+b j \right ) x} dx \right ] 
\hat \psi (p') dp' \nonumber \\
&=& i \sum_{j\in \Z} \int_{\tilde \R} \alpha_j (p,p') \hat \psi (p') \delta (p'-p+bj) d p' \label {Eqn26} \\ 
& =& \sum_{j\in \Z} C_j (p) \hat \psi (p-jb) \nonumber
\ee
because $\overline {u (x,p)}  \frac {\partial u (x,p')}{\partial p}$ is a periodic function with respect to $x$ with period $a$ and 
where we set
\bee
\alpha_j (p,p') = \frac 1a \int_0^a e^{-i  b j x } \overline {u(x,p)} \frac {\partial u (x,p')}{\partial p} dx , \ \ b = \frac {2\pi}{a}\, , 
\eee
and
\bee
C_j (p) :=\alpha_j (p,p-jb)= \frac ia \int_0^a e^{-ibjx} \overline {u (x,p)} \frac {\partial u (x,p-bj)}{\partial p} d x \,  . 
\eee
In fact, we only have to justify the exchange of integration from (\ref {Eqn24}) to (\ref {Eqn25}) and the exchange of the 
operation of sum, with respect to $j$, and integration from (\ref {Eqn25}) to (\ref {Eqn26}). \ To this end we make 
use of the fact that $\hat \psi$ is a rapidly decreasing function, from a standard mollifier argument and from the following technical 
result concerning the deceasing behavior of $C_j (p)$ with respect of $j$ and $p$.

\begin {lemma}\label {Lemma4}
Let the periodic potential regular enough, that is $V \in C^r$ for some $r>0$, then there exists a positive constant $C:=C_r$ 
independent of $j$ and $p$ such that
\bee
\left | C_j (p) \right | \le C [|p|+1]^{-2} [|j|+1]^{-r} \, , \ \forall p \in \tilde \R \ \mbox { and } \ j\in \Z \, . 
\eee
\end {lemma}

\begin {proof}
First of all we remark that from the results obtained in Theorem \ref {Lemma1} it follows that (let us denote $u(x,E) = u \left [ x , E(p) 
\right ]$ instead of $u(x,p)$ where this does not cause misunderstanding)
\bee
C_j (p) 
&=& \frac ia \int_0^a e^{-i b j x } \overline {u (x,p)} \frac {\partial u(x,p-bj)}{\partial p} d x \\ 
&=& \frac ia \int_0^a e^{-i b j x } \left. \overline {u (x,E)} \right |_{E=E(p)} \left. \frac {\partial u(x,E)}{\partial E} 
\right |_{E=E(p-bj)} \frac {d E(p-bj)}{dp} d x 
\eee
is uniformly bounded for any $j$ and any $p$:
\bee
|C_j (p)| \le C \, ,\  \forall j \in \Z, \ \forall p \in \tilde \R \, , 
\eee
for some positive constant $C>0$. \ Furthermore, equation (\ref {Eqn1}) can be rewritten as 
\bee
\delta_j^0 = \frac 1a \int_0^a \overline {\varphi (x,p)} \varphi (x, p - bj) dx = \frac 1a \int_0^a \overline {u (x,p)} u (x, p - bj) 
e^{-ibjx} dx
\eee
then it follows that 
\be
C_j (p) 
= - \frac ia \int_0^a e^{-ibjx} u (x,p-bj) \frac {\partial \overline {u (x,p)}}{\partial p} d x = \overline {C_{-j} (p-bj)} \, . \label {Eqn27}
\ee
Then, we estimate the terms $C_j (p+bj)$ by making use of the asymptotic estimate for $j$ and $p$ large. \ To this end let $E=E(p)$ 
and $E_j =E(p+bj)$; then, from Theorem \ref {Lemma1}, recalling that $E(p) \sim p^2$ and integrating by parts $r$ times it follows that
\bee
&& C_j (p+bj) = \frac {i}{a} \int_0^a e^{-ibjx} \overline {u (x, E_j)} \frac {\partial u(x,E)}{\partial E } \frac {dE}{dp} dx \\ 
&& \ \sim  \frac {2ip}{a^2} \int_0^a e^{-ijbx} \left [ 1 -i \frac {Q(x)}{2\sqrt {E_j}} \right ] \left [ 1 + O(E_j^{-1} ) \right ]  
i \frac {Q(x)}{4 E^{3/2}}  \left [ 1 + O(E^{-1/2}) \right ]2 \sqrt {E} dx \\ 
&& \ \sim O (p^{-2} j^{-r}) \left [ 1 + O(E_j^{-1} ) \right ] \, \left [ 1 + O(E^{-1/2}) \right ]
\eee
provided that $V \in C^{r}$. 
\end {proof}

Thus Theorem \ref {Theorem2} is completely proved since the set of rapidly decreasing functions is a dense set in $L^2 (\tilde R , dp)$.
\end {proof}

\begin {remark} \label {Remark1}
From (\ref {Eqn27}) it turns out
\bee
C_0 (p) = \overline {C_0(p)}
\eee
is real valued. \ In particular, if the periodic potential is a symmetric function, then $\overline {u(x,p)} = u(-x,p)$ and we have that 
\bee
C_0 (p) &=& \frac {i}{a} \int_{-a/2}^{+a/2} u(-x,p) \frac {\partial u (x,p)}{\partial p} dx = \frac {i}{a} \int_{-a/2}^{+a/2} u(x,p) \frac {\partial u (-x,p)}{\partial p} dx \\ 
&=& \frac {i}{a} \int_{-a/2}^{+a/2} u(x,p) \frac {\partial \overline {u (x,p)}}{\partial p} dx = -\frac {i}{a} \int_{-a/2}^{+a/2} \overline {u(x,p)} \frac {\partial {u (x,p)}}{\partial p} dx = - C_0 (p) 
\eee
Hence $C_0 (p) \equiv 0$ is the case of symmetric periodic potentials.
\end {remark}

Now, we have the following property

\begin {theorem} \label {Theorem3}
Let $V\in C^2$. \ Then the linear operator $X$ from $L^2 (\tilde \R, dp)$ to $L^2 (\tilde \R , dp)$ is bounded.
\end {theorem}

\begin {proof}

Let us fix a real and positive number $M>0$ large enough and let us compute the norm of $X \hat \psi$:
\bee
\| X\hat \psi \|^2_{L^2 (\tilde R, dp)} &=& \int_{\tilde \R} \left | \sum_{j\in \Z} C_j (p) \hat \psi (p-bj) \right |^2 dp 
=  \int_{\tilde \R} \left | \sum_{j\in \Z} C_j (p+bj) \hat \psi (p) \right |^2 dp  \\  
&\le & \int_{\tilde \R} \left [ \sum_{j\in \Z} |C_j (p+bj)|  \right ]^2 |\hat \psi (p)|^2 dp \\ 
& \le & 2 \int_{\tilde \R} \left [ \sum_{j^\star} |C_j (p+bj)|  \right ]^2 |\hat \psi (p)|^2 dp 
+ 2 \int_{\tilde \R} \left [ \sum_{j^{\star \star}} |C_j (p+bj)|  \right ]^2 |\hat \psi (p)|^2 dp 
\eee
where $\sum_{j^\star}$ is the sum for all $j\in \Z$ such that $|p+bj|>M$ and $|j|>M$; $\sum_{j^{\star \star}}$ is the sum 
for all $j\in \Z$ such that $|p+bj|\le M$ or $|j| \le M$. \ The second integral is simply estimated as follow
\bee
\int_{\tilde \R} \left [ \sum_{j^{\star \star}} |C_j (p+bj)|  \right ]^2 |\hat \psi (p)|^2 dp \le C^2 M^2 \| 
\hat \psi \|^2_{L^2 (\tilde R, dp)} \, . 
\eee
since from Lemma \ref {Lemma4} it follows that $|C_j (p+bj)| \le C$ for some positive constant $C$. \ Concerning the first 
integral we estimate the terms $C_j (p+bj)$ by making use again of Lemma \ref {Lemma4} obtaining that 
\bee
\int_{\tilde \R} \left [ \sum_{j^{\star }} |C_j (p+bj)| \right ]^2 |\hat \psi (p)|^2 dp \ &\le & \ \int_{\tilde \R} \left [ \sum_{j^{\star}} C [|p+j|^2+1]^{-2} |j|^{-r} \right ]^2 |\hat \psi (p)|^2 dp \\ 
& \le & C \| \hat \psi \|_{L^2 (\tilde R, dp)}^2
\eee
provided that $r\ge 2$. \ Thus, the boundedness of the linear operator $X$ follows. 
\end {proof}

\begin {remark} \label {Remark2}
In fact, the boundedness of the interband operator has been previously proved for potentials such that the associated Bloch operator 
has spectrum with only a finite number of open gaps \cite {GMS} or, in the opposite situation, for potentials such that all 
the gaps are open and with width bigger than a positive constant $C$, for some $C>0$ fixed (see \cite {MS} in the case of periodic 
potential given by a periodic sequence of $\delta'$ distributions). \ With Theorem \ref {Theorem3} we eventually fill that gap between 
these two opposite situations.  
\end {remark}

\section {Analytic distortion of the Stark-Wannier operator in the ECMR}  

Since translation in the $x-$space corresponds, by means of the Fourier transformation, to multiply by and 
exponential function in the momentum representation then we adopt the following strategy: we consider the Stark-Wannier 
operator $\tilde H_F$ in the ECMR  and its analytic deformation consists by multiplying the vector $\hat \psi$ by $e^{i\theta p}$, 
for $p $ outside a given interval. 

Let $N \ge 1$ be fixed and such that the gap between the $N$-th band and the $(N+1)$-th band is not empty: $E^b_{N+1} - E^t_N >0$. \ Let 
$\Gamma_1 = \left ( - \frac 12 N b , + \frac 12 N b \right ) \cap \tilde \R$ and $\Gamma_c = \tilde \R \setminus \Gamma_1$. \ Then we define 

\begin {definition} \label {Def1}
Let ${\mathcal H}_E$ be the space of vectors $\hat \psi \in L^2 (\tilde \R)$ satisfying the periodic boundary conditions at 
$\pm \frac 12 N b$; that is 
\be
\frac {d^r \hat \psi \left ( - \frac 12 N b \pm 0 \right )}{d p^r} = \frac {d^r \hat \psi \left ( + \frac 12 N 
b \mp 0 \right )}{d p^r} \, , \ \forall r \in \N \, . \label {22Bis}
\ee
Hence, ${\mathcal H}_E$ can be identified with $L^2 (\Gamma_1 )\otimes L^2 (\Gamma_c)$ and $\hat \psi$ with a couple of vectors 
$(\hat \psi_1 , \hat \psi_c)$ such that $\hat \psi_1 \left ( - \frac 12 N b +0 \right ) = \hat \psi_1 \left ( + \frac 12 N b -0 \right ) $ 
and $\hat \psi_c \left ( - \frac 12 N b -0 \right ) = \hat \psi_c \left ( + \frac 12 N b +0 \right ) $ together with their 
derivatives. \ The vector $(\hat \psi_1 , \hat \psi_c) = U_E \psi $ will be denoted the EMCR of $\psi$ and the unitary transformation 
from $L^2(\R , dx)$ to ${\mathcal H}_E$ acts as 
\bee
\hat \psi_1 (p) = \int_{\R} \overline {\varphi (x,p)} \psi (x) dx \, , p \in \Gamma_1 \, , \ \hat \psi_c (p) = \int_{\R} 
\overline {\varphi (x,p)} \psi (x) dx \, , p \in \Gamma_c
\eee
with inverse
\bee
\psi (x) = \frac {1}{2\pi} \int_{\Gamma_1} \varphi (x,p) \hat \psi_1 (p) d p + 
\frac {1}{2\pi} \int_{\Gamma_c}  \varphi (x,p) \hat \psi_c (p) d p \, . 
\eee
\end {definition}

We define now the analytic distortion.

\begin {definition} \label {Def2}
Let $\theta \in \C$ be fixed and such that $-R \le \Im \theta \le 0$, we define the analytic deformation  
\bee
\hat \psi^\theta = (\hat \psi_1^\theta , \hat \psi_c^\theta )= \tilde {\mathcal U}^\theta \hat \psi = U_E^\theta \psi \, , \ 
U_E^\theta :=  {\mathcal U}^\theta U_E
\eee
as follows 
\bee
\left \{ 
\begin {array}{l} 
\hat \psi_1^\theta (p) = \hat \psi_1 (p) \\ 
 \hat \psi_c^\theta (p) = e^{ip\theta} \hat \psi_c (p) 
\end {array} \right. \, , 
\eee
where $\hat \psi_1$ and $\hat \psi_c$ are defined in Definition \ref {Def1}, with inverse
\bee
\psi (x) &=& \frac {1}{2\pi} \int_{\Gamma_1} \varphi (x,p) \hat \psi^\theta_1 (p) d p + 
\frac {1}{2\pi} \int_{\Gamma_c} e^{-i p \theta} \varphi (x,p) \hat \psi^\theta_c (p) d p \, , 
\eee
$U_E^\theta$ is an unitary transformation for real $\theta$.

We define the analytic distortion of the Stark-Wannier operator in the ECMR as follows
\be
\tilde H_F^\theta =  U^\theta_E H_F \left ( U^\theta_E \right )^{-1} \, . \label {Eqn29}
\ee
\end {definition}

Now, we are going to give an explicit expression of $\tilde H_F^\theta$. \ Let $\hat \psi^\theta$ be a test function such that 
$e^{-ip\theta} \hat \psi_c^\theta (p)$ is a rapidly decreasing function (where the set of such functions is dense in ${\mathcal H}_E$), 
then 
\bee
\left [H_F \left ( U^\theta_E \right )^{-1} \hat \psi^\theta \right ] (x)&=& \frac {1}{2\pi} \int_{\Gamma_1} \left [ E (p) \varphi (x,p) + 
F x \varphi (x,p) \right ] \hat \psi_1^\theta (p) dp  +  \\ 
& + & \frac {1}{2\pi} \int_{\Gamma_c} \left [ E (p) \varphi (x,p) + F x \varphi (x,p) \right ] e^{-ip\theta} \hat \psi^\theta_c (p) dp \, . 
\eee
In particular, we have that 
\bee
&& \int_{\Gamma_c} x \varphi (x,p)  e^{-ip\theta}  \hat \psi_c^\theta (p) dp =  \int_{\Gamma_c} x e^{i p x} u (x,p) e^{-ip\theta} 
\hat \psi_c^\theta (p) dp = \\ 
&& \  = \int_{\Gamma_c} \frac {\partial  \left ( -i e^{i p x} \right )}{\partial p}u (x,p) e^{-ip\theta} \hat \psi_c^\theta (p) dp \\ 
&& \  =  \left [ -i \varphi (x,p)  e^{-ip\theta} \hat \psi_c^\theta (p) \right ]_{\partial \Gamma_c} + i  
\int_{\Gamma_c}  e^{i p x}  \left [ u (x,p) \frac {\partial e^{-i p \theta} \hat \psi_c^\theta (p)}{\partial p}  + 
\frac {\partial u (x,p)}{\partial p} e^{-i p \theta} \hat \psi_c^\theta (p) \right ] dp \\ 
&& \ \ =   i  \int_{\Gamma_c} \varphi (x,p) \frac {\partial e^{-i p \theta} \hat \psi_c^\theta (p)}{\partial p}  dp + 
i  \int_{\Gamma_c}  e^{i p x}  \frac {\partial u (x,p)}{\partial p} e^{-i p \theta} \hat \psi_c^\theta (p) d p \\ 
&& \ \ = i  \int_{\Gamma_c} \left [ - i \theta \hat \psi_c^\theta (p) + \frac {\partial \hat \psi_c^\theta (p)}{\partial p} \right ]  
e^{-i\theta p}\varphi (x,p) dp + i  \int_{\Gamma_c}  e^{i p x}  \frac {\partial u (x,p)}{\partial p} e^{-i p \theta} \hat \psi^\theta_c (p)dp 
\eee
because both functions $e^{-i p \theta} \hat \psi_c^\theta (p)$ and $\varphi (x,p)$ have the same value at the boundary points of $\Gamma_c$ 
(and $e^{-ip\theta}\hat \psi_c^\theta (p)$ goes to zero as $p$ goes to infinity). \ Similarly, we treat the integral over the set 
$\Gamma_1$. \ Hence,
\bee
\left [ \tilde H_F^\theta \hat \psi^\theta \right ] (p) = 
\left \{
\begin {array}{ll} 
E (p) \hat \psi_1^\theta (p) + i F \frac {\partial \hat \psi_1^\theta (p)}{\partial p} +  F h_1^\theta (p)  
& \, , \ \mbox { if } \ p \in \Gamma_1 \\ 
E(p) \hat \psi_c^\theta (p) + i F \frac {\partial \hat \psi_c^\theta  (p)}{\partial p} + F h_c^\theta (p) 
+ F \theta \hat \psi_c^\theta & \, , \ \mbox { if } \ p \in \Gamma_c
\end {array}
\right.
\eee
where 
\bee
h_1^\theta (p) &=&  \frac {i}{2\pi } \int_{\R} \overline {\varphi (x,p)} \left \{ \int_{\Gamma_1} e^{i q x}  
\frac {\partial u (x,q)}{\partial q} \hat \psi_1^\theta (q) d q + 
\int_{\Gamma_c} e^{i q x}  \frac {\partial u (x,q)}{\partial q}e^{-iq\theta} \hat \psi_c^\theta (q) d q \right \} dx\, , 
\eee
and 
\bee
h_c^\theta (p) &=&  \frac {ie^{ip\theta}}{2\pi } \int_{\R} \overline {\varphi (x,p)} \left \{ \int_{\Gamma_1} e^{i q x}  
\frac {\partial u (x,q)}{\partial q} \hat \psi_1^\theta (q) d q + \int_{\Gamma_c} e^{i q x}  
\frac {\partial u (x,q)}{\partial q}e^{-iq\theta} \hat \psi_c^\theta (q) d q \right \} dx\, . 
\eee 
By making use of the same arguments as in Theorem \ref {Theorem2} then we can exchange the integration sets obtaining that
\bee
h_1^\theta (p) &=& \int_{\Gamma_1} Y(p,q) \hat \psi_1^\theta (q) dq  + \int_{\Gamma_c} Y(p,q) e^{-iq\theta} \hat \psi_c^\theta (q) 
dq \, ,\ p\in \Gamma_1 \\ 
h_c^\theta (p) &=& e^{ip\theta} \int_{\Gamma_1} Y(p,q) \hat \psi_1^\theta (q) dq + \int_{\Gamma_c} Y(p,q) e^{i(p-q)\theta} 
\hat \psi_c^\theta (q) dq \, , \ p \in \Gamma_c 
\eee
where 
\bee
Y(p,q) = \frac {i}{2\pi} \int_{\R} e^{-i(p-q)x} \overline {u (x,p)} \frac {\partial u (x,q)}{\partial q} dx \, . %\label {f28}
\eee

\begin {remark} \label {Remark3}
We should point out that $Y(p,q) $ is well defined because the $N$-th gap is open; if not 
the periodic boundary conditions of the band and Bloch functions at the points $\pm \frac 12 N b$ would imply 
an \emph {artificial} singularity for the Bloch function at $q = \mp \frac 12 Nb \pm 0$ and thus $Y(p,q)$ and 
$\frac {\partial u (x,q)}{\partial q}$ are not well defined in these points.
\end {remark}

\begin {lemma}\label {Lemma5Bis}
If the periodic potential $V(x)$ is an analytic function in a set containing the strip $|\Im x | \le R$ then the Fourier series of 
the periodic function $\overline {u (x,p)} \frac {\partial u (x,q)}{\partial q}$
\be
\overline {u (x,p)} \frac {\partial u (x,q)}{\partial q} = \sum_{j\in \Z} \alpha_j (p,q) e^{i j b x} \label {Eqn18Bis}
\ee 
has coefficients 
\bee
\alpha_j (p,q)= \frac {1}{a} \int_0^a \overline {u (x ,p)} \frac {\partial u (x , q)}{\partial q} e^{- i j b x} dx 
\eee
satisfying the following estimates
\be
|\alpha_j (p,q) |,\ \left | \frac {\partial \alpha_j (p,q)}{\partial p} \right | , \ \left | \frac {\partial \alpha_j (p,q)}
{\partial q} \right | \le C e^{-|j| b R} \, . \label {Eqn28Bis}
\ee
for some positive constant independent of $j$, $p$ and $q$.
\end {lemma}

\begin {proof}
Let us assume for a moment $j>0$. \ Then, from the Cauchy Theorem it follows that
\bee
\alpha_j (p,q) &=& \frac 1a \int_0^a u(x,-p) \frac {\partial u (x,q)}{\partial q} e^{-i j b x} dx \\ 
&=& e^{-j b R} \frac 1a \int_0^a u(x-iR,-p) \frac {\partial u (x-iR,q)}{\partial q} e^{-i j b x} dx
\eee
where $u(x-iR,-p)$ and $\frac {\partial u (x-iR,q)}{\partial q} = \frac {\partial u (x-iR,E)}{\partial E} \frac {dE}{dq}$ are uniformly bounded since Theorem \ref {Lemma1}. \ That is it follows that 
\bee
|\alpha_j (p,q) | \le C (p,q) e^{-|j| b R}  
\eee
where 
\be
C (p,q) = \max_{- R \le \Im x \le 0 } 
\left | u (x,p) \frac {\partial u(x,q)}{\partial q} \right | \, . \label {Eqn28Ter}
\ee
In fact, from the asymptotic behavior in Theorem \ref {Lemma1} then it follows that 
$\frac {\partial u (x,q)}{\partial q} = \frac {\partial u (x,E)}{\partial E}  \frac {d E}{d q}$, where 
$\left | \frac {d E}{d q} \right | \le C |q|$ and $ \left | \frac {\partial u (x,E)}{\partial E} \right | = O(q^{-3})$; hence, there 
exists a positive constant $C$ such that 
\bee
C (p,q) \le C 
\eee
for any $p$ and $q$. \ In the case $j =-|j| <0$ then we estimate its complex conjugate
\bee
\overline {\alpha_j (p,q)} = \frac 1a \int_0^a u(x,p) \frac {\partial u (x,-q)}{\partial q} e^{-i |j| b x} dx
\eee
by means of the same arguments.

In a similar way we get the estimates of the derivatives of $\alpha_j$. \ With more details, the estimate 
(\ref {Eqn28Bis}) of the derivative 
\bee
\frac {\partial \alpha_j (p,q)}{\partial p} = \frac 1a \int_0^a \frac {\partial \overline {u(x,p)}}{\partial p} 
\frac {\partial u(x,q)}{\partial q} e^{-i j b x} dx 
\eee
immediately follows by means of the same arguments as above. \ In order to get the estimate (\ref {Eqn28Bis}) of the derivative
\bee
\frac {\partial \alpha_j (p,q)}{\partial q} = \frac 1a \int_0^a \overline {u(x,p)}
\frac {\partial^2 u(x,q)}{\partial q^2} e^{-i j b x} dx 
\eee
we have to control the second derivative, with respect to $q$, of $u(x,q)$; it comes by deriving both sides of the asymptotic 
behavior given in Theorem \ref {Lemma1} obtaining 
\bee
\frac {\partial^2 u(x,q)}{\partial q^2} = \left ( \frac {d E}{d q} \right )^2 \frac {\partial^2 u (x,E)}{\partial E^2} \sim q^{-3} \, . 
\eee
The Lemma thus follows. \end {proof} 

From this fact and recalling that the Fourier transform of $e^{i b j x}$ is the Dirac's $\delta$ distribution $\delta (x-jb)$ the it 
follows that 
\bee
Y(p,q) = \sum_{j\in \Z} \alpha_j (p,q)  \delta (p-q-jb) 
\eee
where the coefficients $\alpha_j (p,q)$ satisfy the above estimates. \ Collecting all these facts and remembering that $C_0 (p)$ is a real 
valued function as discussed in Remark \ref {Remark1}, 
we can conclude that

\begin {lemma} \label {Lemma5}
We have that the analytically distorted ECMR of the Stark-Wannier operator (\ref {Eqn29}) is given by $\tilde H^\theta_{F,\eta}$ for $\eta =F$, where 
$\tilde H^\theta_{F,\eta}$ is formally defined on $\hat \psi \in {\mathcal H}_E$ as 
\bee
\left ( \tilde H_{F, \eta}^\theta \hat \psi \right ) (p) = 
\left \{ 
\begin {array}{ll} 
\left [ i F \frac {d }{d p} + E (p) \right ] \hat \psi_1 (p) +  F \left ( K_{11} \hat \psi_1 \right ) (p) + 
\eta  \left ( K_{1c}^\theta \hat \psi_c \right ) (p) & , \ p \in \Gamma_1 \\
\left [ i F \frac {d }{d p} + E (p) + F C_0 (p) \right ] \hat \psi_c (p) + F \theta \hat \psi_c (p) + & \\ 
\ \ + \eta \left ( K_{c1}^\theta \hat \psi_1 \right ) (p)  + \eta  \left ( K_{cc}^\theta \hat \psi_c \right ) (p)& , \ p \in \Gamma_c 
\end {array}
\right.
\eee
where $K_{11}$ is the intraband interaction among the first $N$ bands defined as 
\be
\left ( K_{11} \hat \psi_1 \right ) (p) &=& \sum_{j\in \Z \ : \ p-bj \in \Gamma_1} \alpha_j (p,p-bj) \hat 
\psi_1 (p-bj) \, , \ p \in \Gamma_1 \, , \label {Eqn31} 
\ee
$K_{1c}$ and $K_{c1}$ are the interband term between the first $N$ bands and the remainder bands defined as 
\bee 
\left (K_{1c}^\theta \hat \psi_c \right ) (p) &=&  \sum_{j\in \Z \ : \ p-bj \in \Gamma_c} \alpha_j (p,p-bj) e^{-i(p-bj)\theta} 
\hat \psi_c (p-bj) \, , \ p \in \Gamma_1 \\ 
\left (K_{c1}^\theta \hat \psi_1 \right ) (p) &=& \sum_{j\in \Z \ : \ p-bj \in \Gamma_1} \alpha_j (p,p-bj) e^{ip\theta} 
\hat \psi_1 (p-bj) \, , \ p \in \Gamma_c \\ 
\eee
and finally where $K_{cc}^\theta$ is defined as 
\bee
\left (K_{cc}^\theta \hat \psi_c \right ) (p) &=& \sum_{j\in \Z \setminus \{ 0 \} \ : \ p-bj \in \Gamma_c} \alpha_j (p,p-bj) 
e^{-ijb\theta} \hat \psi_c (p-bj) \, , \ p \in \Gamma_c \, . 
\eee
\end {lemma}

In conclusion, we are able to obtain the explicit expression of the analytically distorted ECMR of the Stark-Wannier operator preparing the 
ground in order to apply the regular perturbation theory. \ In fact, by construction it turns out that $\tilde H_{F,0}^\theta$ is the 
decoupled band approximation where the coupling term between the first $N$ bands and the other bands is neglected. 

\section {Spectral analysis of the unperturbed operator $\tilde H^\theta_{F,0}$}

Hereafter, let us fix $N=1$ for the sake of simplicity. \ We introduce a change of notation: in the case of $N=1$ then $\Gamma_1$ simply 
coincides with the Brillouin zone ${\mathcal B}$. \ The vectors $\hat \psi \in {\mathcal H}_E$ can be written 
as $\hat \psi = (\hat \psi_1 (k) , \hat \psi_c (p))$ where $\hat \psi_1 (k)$ satisfies the periodic boundary conditions on the Brillouin 
zone ${\mathcal B}$ and where $\hat \psi_c (p)$ is such that $\hat \psi_c (-\frac 12 b - 0) = \hat \psi_c (\frac 12 b + 0)$ together with 
its derivatives. \ The integral operator $K_{11}$ defined in (\ref {Eqn31}) simply reduces to the multiplication operator $X_{1,1}$ 
defined in (\ref {Eqn22}). 

The unperturbed operator $\tilde H_{F,0}^\theta$ can be written on ${\mathcal H}_E$ 
as $\tilde H_1 \otimes \tilde H_c^\theta$, where $\tilde H_1$ and $\tilde H_c^\theta$ are formally defined 
on $L^2 ({\mathcal B},dk)$ and $ L^2 (\Gamma_c ,dp) $ as 
\bee
\tilde H_1 &=& i F \frac {\partial}{\partial k} + E_1 (k) + F X_{1,1} (k)  \\ 
\tilde H_c^\theta &=& i F \frac {\partial}{\partial p} + E(p) + F C_0 (p) + F \theta 
\eee
and its spectrum is given by the union of the spectrum of the operators $\tilde H_1$ and $\tilde H_c^\theta$.

\begin {lemma} \label {Lemma6}
The resolvent of $\tilde H_1$ on a vector $\phi_1 \in L^2 ({\mathcal B}, dk)$ is given by
\be
\left ( \left [ \tilde H_1 - z \right ]^{-1} \phi_1 \right ) (k) = - 
\frac iF \left \{ \int_{-b/2}^k e^{-\frac iF \int_k^q Q_1 (s) d s } \phi_1 (q) d q - 
\frac {\int_{\mathcal B} e^{-\frac iF \int_k^q Q_1 (s)  d s } \phi_1 (q) dq}{1-e^{\frac iF b (z-{\mathcal E_{1,0}}) }}  
\right \} \label {Eqn32}
\ee
where we set $Q_1(k) = E_1 (k) + F X_{1,1} (k)-z$. \ The spectrum of the operator $\tilde H_1$ on $L^2 ({\mathcal B}, dk)$ is purely discrete 
and it consists of a ladder of simple eigenvalues 
\be
{\mathcal E}_{1,j} =  \frac 1b \int_{\mathcal B} \left [ E_1 (k) + F X_{1,1} (k) \right ] d k + j F a , \ j \in \Z \, . \label {Eqn32Ter}
\ee
\end {lemma}

\begin {proof} Indeed, formula (\ref {Eqn32}) for the resolvent directly comes by means of a simple calculation, and the  
eigenvalue equation has normalized eigenvector
\be
\hat \psi_{1,j} (k) = \frac {1}{\sqrt {b}} e^{-\frac {i}{F} \int_{-b/2}^k \left [ E_1 (q) + F X_{1,1,} (q) - 
E \right ] d q } \label {Eqn32Bis}
\ee
where $E= {\mathcal E}_{1,j} $ in order to have $\hat \psi_{1,j} (-b/2) = \hat \psi_{1,j} (+b/2)$. 
\end {proof}

\begin {lemma} \label {Lemma7} 
The spectrum of the operator $\tilde H_c^\theta$ is purely essential and it coincides with the line $\Im z = F \Im \theta$.
\end {lemma}

\begin {proof} Now, we look for the resolvent of $\tilde H_c^\theta$, let $\phi \in L^2 (\Gamma_c ,dp)$ and we consider the solution of the 
equation (see, e.g., \cite {K}) 
\bee
\left [ i F \frac {\partial}{\partial p} + E(p) + F C_0 (p) + F\theta - z \right ] \hat \psi_c = \phi
\eee
for $p \in \Gamma_c$ and with the condition 
\bee
\hat \psi_c \left (-\frac 12 b-0 \right )= \hat \psi_c \left ( +\frac 12 b+0 \right ).
\eee
If we denote by 
\bee
\Gamma_+ = \tilde \R \cap \left [ \frac 12 b , + \infty \right ) \ \mbox { and } \ \Gamma_- = \left ( - \infty , -\frac 12 b \right ] 
\cap \tilde \R
\eee
then, by means of a direct calculation it turns out that (hereafter $\fint_A$ denotes the integral of a meromorphic function on the 
set $A\tilde \C$)
\bee
\hat \psi_c (p) = 
\left \{
\begin {array}{ll} 
-\frac iF \fint_p^{+\infty} e^{\frac iF \fint_q^p Q_c(s) ds } \phi (q) dq & \ \mbox { if } \ \Im (z-F \theta ) >0 \\ 
\frac iF \fint^p_{b/2} e^{\frac iF \fint_q^p Q_c(s) ds } \phi (q) dq  + C_+ e^{\frac iF \fint_{b/2}^p Q_c(q) dq }  
& \ \mbox { if } \ \Im (z-F \theta ) <0 
\end {array}
\right. \, , \ p \in \Gamma_+
\eee
and
\bee
\hat \psi_c (p) = 
\left \{
\begin {array}{ll} 
\frac iF \fint_p^{-b/2} e^{\frac iF \int_q^p Q_c(s) ds } \phi (q) dq + C_- e^{\frac iF \fint_{-b/2}^p Q_c (q) dq } 
& \ \mbox { if } \ \Im (z-F \theta ) >0 \\ 
-\frac iF \fint_{-\infty}^p e^{\frac iF \fint_q^p Q_c (s) ds } \phi (q) dq  & \ \mbox { if } \ \Im (z-F \theta ) <0 
\end {array}
\right. \, , \ p \in \Gamma_-
\eee
where we set here 
\bee
Q_c (p)=E(p)+F\theta + F C_0 (p) -z \, .
\eee
Now, we have to require that $\hat \psi_c (- b/2-0) = \hat \psi_c (+b/2+0)$, which implies that
\bee
C_+ &=& - \frac iF \fint^{-b/2}_{-\infty} e^{-\frac iF \fint_{-b/2}^q Q_c (s) ds } \phi (q) dq \\ 
C_- &=& - \frac iF \fint_{b/2}^{+\infty} e^{-\frac iF \fint_{b/2}^q Q_c (s) ds } \phi (q) dq 
\eee
Then it follows that the spectrum of $\tilde H_c^\theta$ coincides with the line $\R +i F \Im \theta$ and thus it is purely 
essential. \end {proof}

In conclusion
\bee
\sigma (\tilde H_{F,0}^\theta ) = \sigma_d (\tilde H_{F,0}^\theta ) \cup \sigma_{ess} (\tilde H_{F,0}^\theta ) 
\eee
where
\bee
\sigma_d (\tilde H_{F,0}^\theta ) &=& \sigma (\tilde H_1 ) = \left \{ {\mathcal E}_{1,j}, \ j\in \Z \ \right \} \\ 
\sigma_{ess} (\tilde H_{F,0}^\theta ) &=& \sigma (\tilde H_c^\theta ) = \R +i F \Im \theta
\eee

Now, let us fix $\Im \theta < 0$ and let $\gamma_j = \partial \Omega_j$ where 
\bee
\Omega_j =\left [ {\mathcal E}_{1,j} - \frac 12 F a , {\mathcal E}_{1,j} + \frac 12 F a \right ] 
\times i \left [ \frac 12 F \Im \theta , - \frac 12 F \Im \theta \right ] 
\eee
The following estimates for the resolvent operators of $\tilde H_1$ and $\tilde H_c^\theta$ for $z\in \gamma_j$ hold true

\begin {lemma} \label {Lemma8}
There exists a positive constant $C>0$ independent of $F$ and $j$ such that 
\be
\| [\tilde H_1 - z ]^{-1} \| \le \frac CF ,  \ \ \forall z \in \gamma_j \label {Eqn33}
\ee
and
\be
\| [\tilde H_c^\theta - z ]^{-1} \| \le \frac CF ,  \ \ \forall z \in \gamma_j \, . \label {Eqn34}
\ee
\end {lemma}

\begin {proof} First of all we remark that the Stark-Wannier spectral problem is invariant with respect the 
translation $E \to E+Fa$, then the above estimates, if hold true, are independent of $j$ and thus we can assume, 
for argument's sake, $j=0$. \ The proof of the estimate (\ref {Eqn33}) directly comes from formula (\ref {Eqn32}) and from 
the following estimate
\bee
\min_{z\in \gamma_0} |1-e^{\frac iF b (z-{\mathcal E_{1,0}})} | \ge C >0 \, 
\eee
for some positive constant $C$ independent of $F$. \ In order to prove estimate (\ref {Eqn34}) we observe that $\Im (z-F\theta) >0$ for 
any $z \in \gamma_0$; then from the proof of Lemma \ref {Lemma7} it follows that the resolvent $\hat \psi = [\tilde H_c^\theta - z ]^{-1} 
\phi$, where $\phi \in L^2 (\Gamma_c ,dp)$ is a rapidly decreasing test function, is given by 
\be
\hat \psi (p) = 
\left \{
\begin {array}{ll} 
-\frac iF \fint_p^{+\infty} e^{\frac iF \fint_q^p Q_c(s) ds } \phi (q) dq & \ \mbox { if } \ p \in \Gamma_+  \\ 
\frac iF \fint_p^{-b/2} e^{\frac iF \fint_q^p Q_c(s) ds } \phi (q) dq + & \\ 
\ \ - \frac iF e^{\frac iF \fint_{-b/2}^p Q_c (q) dq } \fint_{b/2}^{+\infty} 
e^{-\frac iF \fint_{b/2}^q Q_c (s) ds } \phi (q) dq  
&  \ \mbox { if } \ p \in \Gamma_- 
\end {array}
\right. \, . \label {Eqn35}
\ee
For $p \in \Gamma_+$  we have that
\be
|\hat \psi (p) | \le \frac {1}{F} \fint_p^{+\infty} e^{c (p-q)} |\phi (q) | dq \, , \ c = \frac {1}{F} \Im (z-F \theta ) \ge a  
\label {Eqn31Bis}
\ee
for some $a>0$ independent of $F$ since $z \in \gamma_0$. \ Then, it turns out that 
\bee
|\hat \psi (p) | \le \frac 1F \| \phi \|_{L^\infty (\Gamma_+)} \int_p^{+\infty} e^{c (p-q)}  dq \le \frac {C}{F} \| \phi 
\|_{L^\infty (\Gamma_+)} 
\eee
for some $C>0$. \ Furthermore
\bee
\fint_{b/2}^{+\infty} |\hat \psi (p) | d p 
& \le & \frac {1}{F} \fint_{b/2}^{+\infty} dp \fint_p^{+\infty} e^{c (p-q)} |\phi (q) | dq \\ 
& \le & \frac {1}{F} \fint_{b/2}^{+\infty} |\phi (q) | dq \int_{b/2}^p e^{c (p-q)}  dp \\ 
& = & \frac {1}{\Im (z-F\theta)} \fint_{b/2}^{+\infty} \left [ 1-e^{- c(q-b/2)} \right ] |\phi (q) | dq \\
&\le & \frac {C}{F} \| \phi \|_{L^1 (\Gamma_+)}  
\eee
for some $C>0$. \ Then, we have that 
\bee
\| \hat \psi \|_{L^\infty (\Gamma_+)} \le \frac {C}{F} \| \phi \|_{L^\infty (\Gamma_+)} \ \ \mbox { and } \ \ 
\| \hat \psi \|_{L^1 (\Gamma_+)} \le \frac {C}{F} \| \phi \|_{L^1 (\Gamma_+)}
\eee
for some $C>0$. \ From this fact and from the Riesz-Thorin interpolation theorem it follows that 
\bee
\| \hat \psi \|_{L^2 (\Gamma_+)} \le \frac {C}{F} \| \phi \|_{L^2 (\Gamma_+)} \, . 
\eee
Similarly, we have the same result for $p\in \Gamma_-$; that is:
\bee
|\hat \psi (p) | \le \frac {1}{F} \left \{ \fint_p^{-b/2} e^{c(p-q)} |\phi (q)| dq + e^{c(p+b/2)} \fint_{b/2}^{+\infty} e^{-c(q-b/2)} 
|\phi (q)| dq \right \} \, , p \in \Gamma_-
\eee
where the first integral can be estimated as in the case $p\in \Gamma_+$, while the second integral is simply estimated as follows
\bee
\left \| e^{c(p+b/2)} \fint_{b/2}^{+\infty} e^{-c (q-b/2) } |\phi (q) | dq \right \|_{L^2(\Gamma_-)} &\le & C \left | \fint_{b/2}^{+\infty} 
e^{-c (q-b/2) } |\phi (q) | dq \right | \\ 
&\le & C \left \| e^{-c (q-b/2) } \right \|_{L^2(\Gamma_+ , dq)} \left \| \phi (q) \right \|_{L^2(\Gamma_+ , dq)} \\ 
&\le & C  \left \| \phi (q) \right \|_{L^2(\Gamma_+ , dq)}
\eee
for some $C>0$. \ Hence, we can conclude that
\bee
\| \hat \psi \|^2_{L^2 (\Gamma_c )} &=& \| \hat \psi \|^2_{L^2 (\Gamma_+ )} + \| \hat \psi \|^2_{L^2 (\Gamma_- )} \\ 
&\le & \frac {C}{F^2} \| \phi \|^2_{L^2 (\Gamma_+ )} + \frac {C}{F^2} \left [ \| \phi \|_{L^2 (\Gamma_- )}+  \| \phi \|_{L^2 (\Gamma_+ )} \right ]^2 \\ 
&\le & \frac {C}{F^2} \left [ \| \phi \|^2_{L^2 (\Gamma_+ )} + \| \phi \|^2_{L^2 (\Gamma_- )} \right ] \\ 
& = & \frac {C}{F^2}  \| \phi \|^2_{L^2 (\Gamma_c )}
\eee
Then, the Lemma follows because the set of rapidly decreasing functions belonging to $L^2 (\Gamma_c , dp)$ is dense in 
$L^2 (\Gamma_c , dp)$. 
\end {proof}

\section {Perturbation results and proof of Theorem \ref {Theorem1}.}

\subsection {Norm estimate of the perturbation}

We then denote by $K^\theta$ the perturbation defined on ${\mathcal H}_E$ as 
\bee
K^\theta \hat \psi = \left ( K^\theta_{1c} \hat \psi_c , K^\theta_{c1} \hat \psi_1 + K^\theta_{cc} \hat \psi_c \right ) 
\eee
where $K^\theta_{1c}$, $K^\theta_{c1}$ and $K^\theta_{cc}$ are defined in Lemma \ref {Lemma5}, and thus 
\bee
\tilde H^\theta_{F,\eta} = \tilde H_{F,0}^\theta + \eta K^\theta \, . 
\eee

\begin {lemma} \label {Lemma9}
Let $V(x)$ be an analytic function in an open set containing the strip $|\Im \theta | \le R$. \ Then, the linear operator $K^\theta$ is 
bounded for any $\theta \in \C$ such that $ 0 < |\Im \theta | < R$:
\bee
\| K^\theta \| \le C_1
\eee
for some $C_1$ independent of $F$.
\end {lemma}

\begin {proof}
The proof basically mimic the proof of Theorem \ref {Theorem3} where, in the case of analytic potential $V(x)$, the exponential 
term $e^{- i j b \theta}$ in $K^\theta_{cc}$ (and similarly the other exponential terms in $K_{1c}^\theta$ and $K_{c1}^\theta$) 
are supported by the exponential estimate (\ref {Eqn28Bis}) since we choose $ |\Im \theta | < R$.
\end {proof}

From this fact and since 
\bee
\left \| \left [ \tilde H^\theta_{F,0} - z \right ]^{-1} \right \| \le \frac {C}{F} , \ \ \forall z \in \gamma_0 
\eee
then it follows that the perturbative series 
\bee
\left [ \tilde H^\theta_{F,\eta} - z \right ]^{-1} = \sum_{j=0}^\infty \left [ \tilde H^\theta_{F,0} - z \right ]^{-1} 
\left ( - \eta K^\theta  \left [ \tilde H^\theta_{F,0} - z \right ]^{-1} \right )^j 
\eee
converges for any $\eta$ such that $|\eta | < C F$ for some $C>0$. \ In fact, we don't have an estimate of such a constant $C$ and thus, in 
order to get results about the spectrum of $\tilde H^\theta_{F,F}$  we need to improve the result of Lemma \ref {Lemma9} and obtain the 
convergence of the perturbation series for $\eta =F$. \ To this end we set 
\bee
{\mathcal K}^\theta (z) = \left [ \tilde H^\theta_{F,0} - z \right ]^{-1}  \eta K^\theta  
\left [ \tilde H^\theta_{F,0} - z \right ]^{-1} \, . 
\eee
and the perturbation series takes the form
\be
\left [ \tilde H^\theta_{F,\eta} - z \right ]^{-1} = \left ( \left [ \tilde H^\theta_{F,0} - z \right ]^{-1} - 
{\mathcal K}^\theta (z) \right )  \sum_{j=0}^\infty \left ( \eta K^\theta {\mathcal K}^\theta (z) \right )^j \label {Eqn36}
\ee

We prove that

\begin {theorem} \label {Theorem3Bis} Let $\theta$ be in a given box $[-\Lambda , + \Lambda ] \times i[-R , 0]$ for some $\Lambda >0$ fixed. \ Then, for $F$ small enough  
\be
\| {\mathcal K}^\theta (z) \| \le C_2 \eta F^{-3/2} \, , \label {Eqn36Bis}
\ee
for any $z\in \gamma_0$ and for some $C_2$ independent of $F$ and $\theta$.
\end {theorem}

\begin {proof} The vector ${\mathcal K}^\theta \hat \psi$ may be written as $({\mathcal K}^\theta_{1c} \hat \psi_c , 
{\mathcal K}^\theta_{c1} \hat \psi_1 + {\mathcal K}^\theta_{cc} \hat \psi_c )$ where 
\bee
{\mathcal K}_{1c}^\theta &=& \eta [\tilde H_1 - z ]^{-1} K_{1,c}^\theta [\tilde H_c^\theta - z]^{-1} \\ 
{\mathcal K}_{c1}^\theta &=& \eta [\tilde H_c^\theta - z ]^{-1} K_{c,1}^\theta [\tilde H_1 - z]^{-1} \\ 
{\mathcal K}_{cc}^\theta &=& \eta [\tilde H_c^\theta - z ]^{-1} K_{c,c}^\theta [\tilde H_c^\theta - z]^{-1} 
\eee
In order to estimate these terms we consider at first ${\mathcal K}_{cc}^\theta \hat \psi_c $. \ Let assume for a moment that 
$p \in \Gamma_+$, then by (\ref {Eqn35}) it follows that:  
\bee
\left [ {\mathcal K}_{cc}^\theta  \hat \psi_c \right ] (p) = - \frac {i\eta}{F} \fint_p^{+\infty} e^{\frac {i}{F} \fint_q^p 
Q_c (s) ds } \left ( K_{cc}^\theta [\tilde H_c^\theta - z ]^{-1} \hat \psi_c \right ) (q) dq = g_+ (p) + g_- (p) 
\eee
where by (\ref {Eqn28Ter})
\bee
g_{\pm} (p)= -i \frac {\eta}{F} \fint_p^{+\infty} e^{\frac {i}{F} \fint_q^p Q_c (s) ds }  \fint_{\Gamma_\pm} Y(q,s) 
e^{i(q-s)\theta } \left ( [\tilde H_c^\theta - z ]^{-1} \hat \psi_c \right ) (s) ds \, d q 
\eee
We consider, at first, the integral over $\Gamma_+$ obtaining that (see Lemma \ref {Lemma5})
\bee
g_+ (p) &=& - \frac {\eta}{F^2} \fint_p^{+\infty} e^{\frac {i}{F} \fint_q^p Q_c (s) ds }  
\sum_{j\in \Z \setminus \{ 0\} \ : \ q-jb \in \Gamma_+}  \alpha_j (q,q-jb) e^{ij b \theta } \fint_{q-jb}^{+\infty} 
e^{\frac iF \fint_{\zeta}^{q-jb} Q_c (\xi ) d \xi } \hat \psi_c (\zeta ) d\zeta  \, d q 
\eee
As a first step we should remark that we can exchange integration and sum operations because 
\bee
&& \left | e^{\frac {i}{F} \fint_q^p Q_c (s) ds } \right | 
= \left | e^{-\frac {1}{F} \fint_q^p \Im Q_c (s) ds } \right | = e^{c (p-q)}  \\ 
&& \left | e^{\frac {i}{F} \fint_\zeta^{q-jb} Q_c (s) ds } \right | 
= \left | e^{-\frac {1}{F} \fint_\zeta^{q-jb} \Im Q_c (s) ds } \right | = e^{c (q-jb-\zeta)} 
\eee
for some positive constant $c>0$ defined in (\ref {Eqn31Bis}) and where $p<q$ and $q-jb<\zeta $ in the domain of integration, 
and (see (\ref {Eqn28Bis}) and (\ref {Eqn28Ter})) 
\bee
\left | \alpha_j (q,q-jb) \right | \le C e^{-|j| b (R-\Im \theta ) }\, , \ \Im \theta < R \, , 
\eee
uniformly with respect to $q$. 

Hence $g_+(p)$ may be written as 
\be
g_+ (p) &=& - \frac {\eta}{F^2} \sum_{j\in \Z \setminus \{ 0\} } e^{ijb\theta} \int_{p-jb}^{+\infty} 
 \hat \psi_c (\zeta ) d \zeta  \times \nonumber \\ 
&& \ \ \times \left [ \int_p^{\zeta + jb} \chi (q-jb-\frac 12 b) \alpha_j (q,q-jb) 
e^{\frac iF \fint_q^p Q_c (s) ds + \frac iF \fint_\zeta^{q-jb} Q_c (s) ds } dq \right ] \nonumber \\ 
&=& - \frac {\eta}{F^2} \sum_{j\in \Z \setminus \{ 0\} } e^{ijb\theta} \fint_{p-jb}^{+\infty} 
  e^{\frac iF \fint_\zeta^p Q_c (s) ds } V_j (p,\zeta ) \hat \psi_c (\zeta ) d \zeta \label {Eqn37}
\ee
where $\chi (x) $ is the Heaviside function (i.e. $\chi (x) =1 $ for any $x\ge 0$ and $\chi (x)=0$ elsewhere), and 
\bee
V_j (p,\zeta ) = \fint_p^{\zeta + jb} \chi (q-jb-\frac 12 b) \alpha_j (q,q-jb) 
e^{-\frac iF \fint_q^{q-jb} Q_c (s) ds } dq \, .
\eee
By means of a stationary phase argument we obtain the following estimate.

\begin {lemma} \label {Lemma11}
We have that  
\bee
|V_j (p, \zeta ) | \le C F^{1/2} |\zeta + j b - p |e^{-c j b - |j| b R}
\eee
for some positive constants $c, C>0$ independent of $j$ and $F$ ($c$ is defined in (\ref {Eqn31Bis})).
\end {lemma}

\begin {proof}
The function $V_j (p,\zeta )$ can be written in the following form
\be
V_j (p, \zeta ) = \fint_{q_1}^{q_2} v_j (q) dq \, , \ v_j (q):= \alpha_j (q, q-jb) e^{-\frac {i}{F} \fint_{q}^{q-jb} E(s) ds } e^{-i \fint_q^{q-jb} 
\left [ C_0 (s) + \theta - z/F \right ] ds }  \label {Eqn34Bis}
\ee
where we set $q_2=\zeta+jb$ and $q_1 = \max \left [ p, \left ( j + \frac 12 \right ) b \right ]$, and where 
$z \sim F$ since $z \in \gamma_0$. \ Hence, the stationary points are the solutions $q \in \tilde \C$ of the equation 
\bee
E(q-jb)=E(q) \, . 
\eee
Recalling that $E(q)$ is a monotone increasing function on $\Gamma$ and that
\bee
E(-p)=E(p) \ \mbox { and } \ E (\bar p ) = \overline {E(p)} 
\eee
then the solution on $\tilde \C$ of the previous equation are only the Kohn branch points $k_j$. \ Hence, we could apply the stationary 
phase method obtaining the estimate
\be
|V_j (p, \zeta )| \le C_j F |\zeta + j b - p | e^{-c j b - |j| b R} \label {Eqn34Ter}
\ee
where $C_j \le C/g_j$ and $g_j$ is the gap between the two bands; indeed, by integration by parts we have that (\ref {Eqn34Bis}) has a 
denominator of the form $E(q-jb)-E(q)$ where
\bee
\min_{q\in \tilde \R} |E(q-jb)-E(q) | = \left | E \left ( \frac 12 j b \pm 0 \right ) - E \left ( - \frac 12 j b \mp 0 \right ) \right | = E_{j+1}^t - E_j^b := g_j 
\eee
In order to have an estimate uniform with respect to $j$ let us fix $\mu \in (0,1)$ and let
\bee
S_1 &=& \{ q \ : \ q\in [q_1,q_2] \cap \tilde \R , \ |E(q-jb)-E(q)| \ge F^\mu \} \\
S_2 &=& \{ q \ : \ q\in [q_1,q_2] \cap \tilde \R , \ |E(q-jb)-E(q)| < F^\mu \}
\eee
and let, with obvious meaning of notation, $V_j (p,\zeta ) = \fint_{S_1} v_j (q) dq + \fint_{S_2} v_j (q) dq$. \ The integral 
over the set $S_1$ may be estimated by integration by parts obtaining that
\bee
\left | \fint_{S_1} v_j (q) dq \right | 
&\le & F (q_2-q_1) e^{-c j b} \max_q \left | \frac {d}{dq} \frac {\alpha_j (q,q-jb) }{E(q-jb)-E(q)} \right | F^{-\mu} \\ 
&\le & C |\zeta + j b - p | F^{1-\mu} e^{-c j b - |j| b R}
\eee
since Lemma \ref {Lemma5Bis}. \ In order to estimate the second integral we simply observe that $E(p)$ has a square branch point 
at the Kohn's branch points and that $E(p) \sim p^2$ for large $p$, thus we have that the measure of the set $S_2$ is of order 
$F^\mu /j$. \ From this fact and from Lemma \ref {Lemma5Bis} again then it follows that
\bee
\left | \fint_{S_2} v_j (q) dq \right | \le C F^{\mu} e^{-c j b - |j| b R}\, . 
\eee
The result follows by choosing $\mu = \frac 12$.
\end {proof}

As a consequence of such a result it follows that 
\bee
|g_+ (p)| &\le & \frac {\eta}{F^2} \sum_j e^{-jb \Im \theta} \fint_{p-jb}^{+\infty} e^{c(p-\zeta )} |V_j (p,\zeta )| \, |\hat \psi_c 
(\zeta ) | d\zeta \\ 
&\le & C \frac {\eta F^{1/2} }{F^2} \sum_j e^{-jb \Im \theta} \fint_{p-jb}^{+\infty} e^{c(p-\zeta )}  |\zeta + j b - p |e^{-c j b - |j| b R} \, |\hat \psi_c 
(\zeta ) | d\zeta \\ 
&\le & C \eta F^{-3/2} \| \hat \psi_c \|_{L^\infty }
\eee
and similarly
\bee
\| g_+ (p) \|_{L^1 (\Gamma_+)} & \le & \frac {\eta}{F^2} \sum_{j\not= 0} e^{-j b \Im \theta} \fint_{\frac 12 b}^{+\infty} dp 
\left | \fint_{p-jb}^{+\infty} e^{c(p-\zeta )} |V_j (p,\zeta )| \, |\hat \psi_c (\zeta ) | d\zeta \right | \\ 
& \le & \frac {\eta}{F^2} \sum_{j\not= 0} e^{-j b \Im \theta} \fint_{\frac 12 b-jb}^{+\infty} |\hat \psi_c (\zeta ) | d\zeta 
\fint_{\frac 12 b}^{\zeta-jb} e^{c(p-\zeta )} |V_j (p,\zeta )| dp \\ 
& \le & \frac {\eta}{F^2} \sum_{j\not= 0} e^{-j b \Im \theta} \fint_{\frac 12 b-jb}^{+\infty} |\hat \psi_c (\zeta ) | d\zeta 
\int_{\frac 12 b}^{\zeta-jb} e^{c(p-\zeta )} C F^{1/2} |\zeta + j b - p |e^{-c j b - |j| b R} dp \\ 
& \le & C \eta F^{-3/2} \| \hat \psi_c \|_{L^1 (\Gamma_c) }
\eee
Since a similar estimate still hols true for the vector $g_- (p)$ and when $p\in \Gamma_-$; then, by means of the Riesz-Thorin interpolation 
argument already applied in the proof of Lemma \ref {Lemma8}, the estimate (\ref {Eqn36Bis}) follows. \ Similarly we have the estimates of 
the other two terms ${\mathcal K}^\theta_{1c} \hat \psi_c$ and ${\mathcal K}^\theta_{c1} \hat \psi_1 $.

The proof of Theorem \ref {Theorem3Bis} is so completed.
\end {proof}

\subsection {Proof of Theorem \ref {Theorem1}}

As a result of Theorem \ref {Theorem3Bis} it follows that the series (\ref {Eqn36}) converges for any $\eta$ such that 
$C_1 C_2 \eta^2 < F^{3/2}$. \ In particular, if $F$ is small enough then the series (\ref {Eqn36}) converges for $\eta =F$. \ We 
have proved that 

\begin {theorem} \label {Theorem4}
Let $\theta$ be fixed and such that $-R < \Im \theta <0$,  let $F$ small enough, then the spectrum of $\tilde H^\theta_F := \tilde H^\theta_{F,F}$ inside 
$\gamma_0$ consists of one, and only, non degenerate eigenvalue ${\mathcal E}_{1,0} (F)$
\end {theorem}

\begin {proof} This result holds true because, from the convergence of the perturbation series (\ref {Eqn36}), it turns out that 
$\gamma_0 \subset \rho \left ( \tilde H^\theta_{F,\eta} \right ) $ for any $\eta \in [0,F]$. \ On the other side, $\tilde 
H^\theta_{F,\eta}$ is an analytic family of operator of type (A). \ Hence the dimension of the range of the projection
\bee
P^\theta (\eta ) = - \frac {1}{2\pi i} \oint_{\gamma_0} [\tilde H^\theta_{F,\eta} - z ]^{-1} dz 
\eee
is constant for any $\eta \in [0,F]$, where dim(Ran($\tilde H^\theta_{F,0}$))=1. \ Hence, dim(Ran($\tilde H^\theta_{F,F}$))=1
\end {proof}

\begin {remark}
Similarly, the perturbation argument gives that $\sigma_{ess} (\tilde H^\theta_{F,F} )$ is contained in the strip $-C F \le \Im z 
\le - \frac 12 |\Im \theta | F$ for some $C>0$. \ In fact, we would expect that the essential spectrum coincides with the line 
$\Im z = - |\Im \theta | F$ as for the unperturbed problem.
\end {remark}

\begin {remark} \label {Remark4}
By construction it turns out that the spectrum of $\tilde H^\theta_{F,F}$ in the strip $|\Im z| < \frac 12 |\Im \theta | F$ is given by a 
ladder of simple eigenvalues
\bee
{\mathcal E}_{1,j} (F) = {\mathcal E}_{1,0} (F) + F j a . 
\eee 
Furthermore, a direct computation gives that the first perturbative term is of order $\asy (F^2)$. \ Indeed, let $\hat \psi_0 = 
(\hat \psi_{1,0} , \hat \psi_{c,0} )$ be the unperturbed eigenvector associated to the eigenvalue ${\mathcal E}_{1,0}$ given in 
(\ref {Eqn32Ter}), where $\psi_{c,0}\equiv 0$ and $\psi_{1,0}$ is given in (\ref {Eqn32Bis}). \ Then, the simple eigenvalue 
${\mathcal E}_{1,0} (F)$ is given by 
\bee
{\mathcal E}_{1,0} (F) = \frac {\langle \phi , \tilde H^\theta_{F,F} P^\theta (F) \hat \psi_{1,0} \rangle}{\langle \phi ,  
P^\theta (F) \hat \psi_{1,0} \rangle}
\eee
where $\phi$ is any vector such that ${\langle \phi ,  P^\theta (F) \hat \psi_{1,0} \rangle} \not= 0$, e.g. $\phi = \hat \psi_{1,0}$. \ In 
particular, an straightforward calculation gives that the first perturbative term is exactly zero: ${\mathcal E}_{1,0} (F) = 
{\mathcal E}_{1,0} + O (F^2)$.
\end {remark}

\begin {remark}
If $N$ is bigger than $1$ then the same arguments apply again obtaining that 
\bee
\mbox {\rm dim}(\mbox {\rm Ran}(\tilde H^\theta_{F,F}))=N
\eee
for $F$ small enough, i.e. $0<F<F_N$ for some $F_N>0$. \ In such a case the spectrum of $\tilde H^\theta_{F,F}$ in the strip $|\Im z| < 
\frac 12 |\Im \theta | F$ is given by $N$ ladders of simple eigenvalues. \ In fact, $F_N$ will depend on the norm of the perturbation and on the distance of the unperturbed eigenvalues from the boundary $\gamma_0$, then we expect that $F_N \to 0$ as $N\to + \infty$.
\end {remark}

We are ready now to conclude the proof of Theorem \ref {Theorem1}. \ Let
\bee
{\mathcal R} := \{ {\mathcal E}_{1,0} (F) + F a j \, , \ j \in \Z \}  = \sigma \left ( \tilde H^\theta_{F,F} \right ) 
\cap \{ z\in \C \ : \ |\Im z | \le \frac 12 |\Im \theta | F \} 
\eee
Furthermore, we remark that $K^\theta$ is a bounded operator-valued analytic function on $|\Im \theta | < R$. \ Hence 
$\tilde H_{F,F}^\theta $ is an analytic family of type (A). \ By means of standard arguments (see \cite {CFKS}) then it follows 
that the eigenvalues of $\tilde H^\theta_{F,F}$ are $\theta$-independent (as long as the essential spectrum does not intersect these 
eigenvalues). \ Hence ${\mathcal R}$ does not depends on $\theta$ for $F$ small enough and 
\bee
{\mathcal R} \subset \{ z\in \C \ : \ -\frac 12 |\Im \theta | F < \Im z < 0 \} \, . 
\eee
Indeed, it follows that $\Im {\mathcal E}_{1,0} (F)  <0$; \ If not then ${\mathcal E}_{1,0} (F)$ is an eigenvalue of $\tilde H^\theta_{F,F}$ 
for any $\theta$, even for $\theta =0$, and then it is an eigenvalue of $H_F$, in contradiction with the fact that the spectrum of $H_F$ 
is purely absolutely continuous. \ The set ${\mathcal R}$ is then the set of quantum resonances; indeed, by making use of the standard 
arguments discussed in \cite {CFKS}, it follows that the function
\be 
\langle \hat \psi , [\tilde H_F - z ]^{-1} \hat \psi \rangle \label {Eqn38}
\ee
has an analytic continuation from $\Im z >0$ to the set $\left \{ \Im z >- \frac 12 F |\Im \theta | \right \} \setminus 
{\mathcal R}$. \ If $z_0 \in {\mathcal R}$ then there exists an $\hat \psi$ such that (\ref {Eqn38}) has an analytic continuation 
to $\Im z >- \frac 12 F |\Im \theta |$ with poles in the points of ${\mathcal R}$.

\appendix

\section {Derivation of (\ref {Eqn20})}\label {AppA} In order to derive formula (\ref {Eqn20}) we make use of the same ideas as 
in \cite {K}. \ If the periodic potential is not a symmetric function the we may write (see \S 4 in \cite {K})
\bee
\varphi (x,E)= \frac {\chi (x,E)}{N^{1/2}(E)} 
\eee
where $N$ is a normalization factor and where
\bee
\chi (x,E) = \phi_1 (x,E) \phi_2 (a,E) - \left [ \phi_1 (a,E) -\lambda (E) \right ] \phi_2 (x,E) \, . 
\eee
Following the same argument by \cite {K} we only have to estimate the integral 
\bee
I := \int_0^a \chi (x,E) \tilde \chi (x,E+\delta E) dx 
\eee
where 
\bee
\tilde \chi (x,E) = \bar \chi (x,E) = \phi_1 (x,E) \phi_2 (a,E) - \left [ \phi_1 (a,E) -\frac {1}{\lambda (E)} \right ] \phi_2 (x,E) \, .
\eee
By making use of the quasiperiodic boundary conditions 
\bee
\left \{
\begin {array}{lcl}
\chi (a,E) &=& \lambda (E) \chi (0,E) \\ 
\frac {\partial \chi (a,E)}{\partial x} &=& \lambda (E) \frac {\partial \chi (0,E)}{\partial x} 
\end {array}
\right.
\ \mbox { and } \ 
\left \{
\begin {array}{lcl}
\tilde \chi (a,E) &=& \frac {1}{\lambda (E)} \tilde \chi (0,E) \\ 
\frac {\partial \tilde \chi (a,E)}{\partial x} &=& \frac {1}{\lambda (E)} \frac {\partial \tilde \chi (0,E)}{\partial x} 
\end {array}
\right.
\eee
a straightforward calculation yields to
\bee
I = - \lambda' (E) \left ( 1- \frac {1}{\lambda^2 (E)} \right ) \phi_2 (a,E) \delta E + O\left ((\delta E)^2\right ) 
\eee
as in the case of symmetric potentials. \ Hence $N(E) =-2 \phi_2 (a,E)\frac {d\mu }{dE}$ holds true even in the case of not symmetric 
potential too, proving (\ref {Eqn20}).

\end {document}